\documentclass[letterpaper]{article}

\usepackage{graphicx} 

\usepackage{authblk} 
\usepackage{fullpage}

\usepackage[round]{natbib}
\usepackage{algorithm}
\usepackage[noend]{algorithmic}

\usepackage{xcolor}
\definecolor{TUMBlue}{HTML}{0065BD}

\usepackage{complexity} 
\usepackage{thm-restate} 
\usepackage{amsmath,amsthm,amssymb}
\usepackage{cleveref}
\usepackage{tikz}
    \usetikzlibrary{calc,decorations.pathreplacing}
\usepackage{thm-restate} 
\usepackage{tabularx, booktabs}
\usepackage{nicefrac}
\allowdisplaybreaks

\newcommand{\publ}[1]{}

\newtheorem{theorem}{Theorem}% 
\newtheorem{proposition}{Proposition}
\newtheorem{lemma}{Lemma}
\newtheorem{corollary}{Corollary}
\newtheorem{claim}{Claim}

\newenvironment{claimproof}
{
	
	\proof
}
{
	\endproof
	
}

\theoremstyle{definition}
\newtheorem{definition}{Definition}

\newcommand{\pr}{\mathbb P} 
\newcommand{\EV}{\mathbb E} 
\newcommand{\orderof}[1]{O\!\left(#1\right)} 
\newcommand{\arr}[1][\sigma]{\prec^{#1}} 
\newcommand{\avai}{\mathcal A}

\newcommand{\matfn}{\rho} 
\newcommand{\genmat}{\mu}

\newcommand{\gdy}{\mathit{GDY}}
\newcommand{\ggdy}{\mathit{DTA}}
\newcommand{\rmat}{\mathit{GMA}}
\newcommand{\wgdy}{\mathit{WGDY}}
\newcommand{\maxe}{\mathit{MAXE}}

\newcommand{\iew}{\mathit{IWA}}
\newcommand{\alg}{\mathit{ALG}}
\newcommand{\ialg}{\mathit{I}\textnormal{-}\mathit{ALG}}

\newcommand{\Umax}{U_{\mathit{max}}}
\newcommand{\Umin}{U_{\mathit{min}}}
\newcommand{\SW}[1][\pi]{\mathcal{SW}(#1)}

\begin{document}

\title{Online Coalition Formation under Random Arrival or Coalition Dissolution}

\author[1]{Martin Bullinger}
\author[2]{Ren{\'e} Romen}

\affil[1]{\small Department of Computer Science, University of Oxford}
\affil[2]{\small School of Computation, Information and Technology, Technical University of Munich
\protect\\ \vspace*{0.05cm}
martin.bullinger@cs.ox.ac.uk, rene.romen@tum.de}

\date{}

\maketitle

\begin{abstract}
Coalition formation explores how to partition a set of $n$ agents into disjoint coalitions according to their preferences.
We consider a cardinal utility model with an additively separable aggregation of preferences and study the online variant of coalition formation, where the agents arrive in sequence.
The goal is to achieve competitive social welfare. 
In the basic model, agents arrive in an arbitrary order and have to be assigned to coalitions immediately and irrevocably. 
There, the natural greedy algorithm is known to achieve an optimal competitive ratio, which heavily relies on the range of utilities.

We complement this result by considering two related models.
First, we study a model where agents arrive in a random order.
We find that the competitive ratio of the greedy algorithm is $\Theta\left(\frac{1}{n^2}\right)$. 
In contrast, an alternative algorithm, which is based on alternating between waiting and greedy phases, can achieve a competitive ratio of $\Theta\left(\frac{1}{n}\right)$.
Second, we relax the irrevocability of decisions by allowing the dissolution of coalitions into singleton coalitions.
We achieve an asymptotically optimal competitive ratio of $\Theta\left(\frac 1n\right)$ by drawing a close connection to a general model of online matching. 
Hence, in both models, we obtain a competitive ratio that removes the unavoidable utility dependencies in the basic model and essentially matches the best possible approximation ratio by polynomial-time algorithms.
\end{abstract}

\section{Introduction}

Coalition formation is a vibrant topic in multi-agent systems.
The goal is to partition a set of $n$ agents into disjoint coalitions.
We consider the framework of hedonic games, where the agents have preferences for the coalitions they are part of by disregarding externalities \citep{DrGr80a}.
More specifically, we assume that each agent has cardinal utilities for each other agent and that utilities for coalitions are aggregated in an additively separable way by taking sums \citep{BoJa02a}.

Most hedonic games literature considers an offline setting, where a fully specified instance is given.
Thus, a good partition has to be computed under full knowledge of the preferences.
By contrast, in many real-life situations, such as the formation of teams in a company, the agents are not all present in the beginning but rather join in over time during an ongoing coalition formation process.
With this motivation in mind, \citet{FMM+21a} proposed an online version of hedonic games, where agents arrive individually.
Upon arrival, an agent reveals her preferences for coalitions containing the agents present thus far and has to be added immediately and irrevocably to an existing coalition (or to a new singleton coalition).
The goal is to find algorithms with high social welfare in a worst-case analysis against an adversary that can select an arbitrary instance and an arbitrary arrival order of the agents.
The performance of an algorithm is measured by its competitive ratio, that is, the worst-case ratio of the social welfare of the computed partition compared to the social welfare of the optimal offline partition.
The main result by \citet{FMM+21a} is that the greedy algorithm, which adds every agent to the best possible coalition upon arrival, has the optimal competitive ratio of $\Theta\left(\frac 1n \frac{\Umin}{\Umax}\right)$, where $\Umax$ and $\Umin$ are the maximum and minimum absolute value of a non-zero single-agent utility in the adversarial instance, respectively.\footnote{\citet{FMM+21a} simplify their model by scaling utilities such that $\Umin$ is always~$1$.}
In other words, even for a fixed number of agents, the competitive ratio is infinite if the maximum or minimum utility is not bounded. 

Arguably, unbounded utilities 
give a lot of power to an adversary.
The following example by \citet{FMM+21a} exploits this.
Consider the situation where the first two arriving agents have a mutual utility of $1$.
Suppose an algorithm adds the second agent to the coalition of the first agent. 
In that case, an adversary can set the utility of $1$ as the minimum utility and add agents with large positive and negative utilities that bring no gain to this coalition.
Hence, the welfare of the large positive utilities is lost, and the algorithm performs poorly.
Otherwise, if an algorithm puts the second agent into a singleton coalition, the adversary can set the utility of other arriving agents to $0$, and the loss of the welfare of the single positive edge leads to an unbounded competitive ratio.

Consequently, we want to see whether we can perform better in two related models.
First, we study the \emph{random arrival model}, where the adversary's power is diminished. 
We still allow the adversary to fix (harmful) utilities. 
However, agents arrive in a uniformly random order. 
We show that the greedy algorithm achieves a competitive ratio of $\Theta\left(\frac{1}{n^2}\right)$.
This is interesting because it shows the capability of removing the utility dependencies in the base model.
Even more, we present an alternative algorithm with a competitive ratio of $\Theta\left(\frac{1}{n}\right)$, which is based on alternating between waiting and greedy phases.
We discuss challenges for obtaining algorithms with an optimal competitive ratio by showing that an optimal stopping algorithm can solve our worst-case instances with a constant-factor welfare loss.

In our second model, we maintain the full power of the adversary to fix both the utilities and the arrival order.
Instead, we enhance the power of algorithms by letting them slightly rearrange coalitions.
During a coalition formation process, it might become apparent that some coalitions do not achieve the desired performance, while rearranging coalitions might lead to large gains.
The basic online coalition formation model by \citet{FMM+21a} does not allow for any changes to coalitions apart from the addition of new agents.
This makes intuitive sense, for instance in a setting of forming teams in a company. 
There, rearranging existing teams 
can lead to performance losses, which may be negligible when merely adding agents to teams.\footnote{Note that our preference model does not capture such performance losses. The social welfare of a partition is rather interpreted as the level of compatibility of an established partition.} 
Still, disallowing for any rearrangement feels unnecessarily prohibitive.
By contrast, we would like to allow small changes to partitions that seem acceptable, even if there exists an aversion to rearrangements.
We, therefore, study algorithms that have the power to completely \emph{dissolve} an existing coalition into singleton coalitions.
As a penalty, we lose all the welfare obtained by this coalition thus far, which can be interpreted as the cost of rearrangement.
At the same time, it opens the possibility of creating new coalitions.
Note that while this allows for reversing earlier bad decisions, it still yields limited power to rearrange coalitions.
In particular, as agents can only be added to a coalition once, the algorithm cannot revoke the decision of \emph{not} adding an agent to a coalition. 
Inspired by the free disposal model in the online matching literature \citep{FKM+09a}, we coin this the \emph{free dissolution model}.

We remark that the free dissolution model also has appeal in coalition formation settings with non-human agents.
Consider the task of classifying a large set of objects arriving over time, say all Wikipedia pages.
Whenever a new page is created, we gain information about its similarity or dissimilarity to existing pages and have to assign it an existing or new category.
The goal is to maximize the total similarity of all categories.
Once we have categorized a large number of objects, restructuring the whole categorization might be infeasible.
However, dividing a single category into subcategories seems more plausible.
New objects can then be assigned to the subcategories.
Importantly, free dissolution requires each object to be assigned its individual subcategory.
Note that imposing this restriction only strengthens our algorithmic possibilities because we already obtain an improved competitive ratio, even with fairly limited power to rearrange.\footnote{
Moreover, if the division into subcategories could be formed arbitrarily, we would lose the essence of an online model: for instance, under the knowledge of the number of objects, we could simply form one large cluster until the last object arrives, and then dissolve to an optimal offline partition.}

For our analysis of the free dissolution model, we draw a close connection to a general model of online matching, where edges are weighted and all agents arrive in a random order.
By transferring results from the edge arrival model (instead of vertex arrival) studied by \citet{GKM+19a} and \citet{Vara11a}, we obtain an algorithm for the matching domain with the optimal competitive ratio of $\frac 1 {3 + 2\sqrt{2}}$.
Now, every matching algorithm is also a coalition formation algorithm: 
it returns a valid partition, which has the property that all coalitions are of size at most~$2$.
Hence,
we can utilize the obtained matching algorithm in the coalition formation domain. 
This leads to a competitive ratio of $\frac 1 {(3 + 2\sqrt{2})n}$.
Interestingly, we then show that this competitive ratio is asymptotically optimal.

Hence, compared to the basic model, we gain compelling new insights.
First, we show that the greedy algorithm does not achieve an optimal competitive ratio anymore.
By contrast, we can be better by a linear factor with more sophisticated algorithms.
Moreover, in both our models, we find algorithms with a competitive ratio of $\Theta\left(\frac 1n\right)$.
This eliminates the utility dependencies of the optimal competitive ratio in the base model.
Notably, all algorithms that we provide are deterministic.
In addition, even though we operate in an online model that is not necessarily bound to limited computational power when performing a decision in each step, all our algorithmic decisions can be performed in polynomial time with respect to the number of agents.
It is known that, for any $\epsilon > 0$, social welfare cannot be approximated in polynomial time by a factor of $\frac 1 {n^{1-\epsilon}}$ \citep{FKV22a}.
Hence, we identify two online settings, where we can efficiently achieve the same approximation guarantee that is essentially possible by polynomial-time offline algorithms.

\section{Related Work}

In this section, we review related work.
We give a detailed account of the literature on coalition formation in the framework of hedonic games as well as research concerning online algorithms with a focus on online matching.

\subsection{Coalition Formation in the Framework of Hedonic Games}

Hedonic games originated from economic theory \citep{DrGr80a} more than four decades ago.
However, their broad consideration only started two decades later \citep{BKS01a,BoJa02a,CeRo01a}.
An overview of hedonic games is provided in the book chapters by \citet{AzSa15a} and \citet{BER24a}.
Unlike in a matching setting, where an agent can only be in a linear number of possible ``coalitions'' of size~$2$, the number of potential coalitions of an agent in a general coalition formation scenario is exponential with respect to the number of agents.
This has led to an abundance of representation formalisms that have to navigate the trade-off of expressivity and succinctness.
Fully expressive models are, e.g., obtained by listing all relevant coalitions \citep{Ball04a} or by representing preferences with Boolean formulas \citep{ElWo09a}.
Such a representation is, however, not possible in polynomial space.
Hence, further effort has been undertaken to identify succinct, but still meaningful preference models.
A prominent idea is to define a utility for a coalition based on a complete and weighted digraph.
The idea is to interpret the weight of outgoing edges as values for single agents, which can be aggregated to values of coalitions.
Taking the sum of weights leads to additively separable hedonic games (ASHGs) \citep{BoJa02a}, while considering the average weight leads to fractional hedonic games (FHGs) \citep{ABB+17a} and modified fractional hedonic games \citep{Olse12a}.

Our paper considers ASHGs and we now focus on related work concerning these games.
A large part of their research focuses on computational issues related to various stability measures \citep{ABS11c,BBT23a,DBHS06a,FMM+21a,GaSa19a,Olse09a,SUDi10a}, but some studies also consider economic efficiency in the sense of Pareto optimality \citep{ABS11c,Bull19a,EFF20a}, popularity \citep{ABS11c,BrBu20a}, or strategyproofness~\citep{WrVo15a,FKMZ21a}.
Maximizing social welfare is a stronger requirement than Pareto optimality and the computational complexity of this problem has been studied in depth \citep{ABS11c,FKV22a,BCS25a}.
In general, maximizing social welfare is \NP-hard, even if valuations are symmetric and  restricted to $\{-1,1\}$ \citep{ABS11c}.
Furthermore, for any $\epsilon > 0$, approximating social welfare by a factor of $\frac 1{n^{1-\epsilon}}$ cannot be done in polynomial time unless \P\,=\,\NP{} \citep{FKV22a}.
This result holds even if valuations are restricted to be symmetric and in $\{-1,0,1\}$ \citep{BCS25a}.
Hence, in contrast to the investigation of online ASHGs by \citet{FMM+21a}, this result does not rely on unbounded utilities.
Nonetheless, \citet{Bull19a} presents a polynomial-time algorithm to achieve the weaker requirement of Pareto optimality for ASHGs that satisfy a weak condition on $0$-utilities.
However, the partitions computed by this algorithm can have negative social welfare and, therefore, cannot achieve any approximation guarantee.

Coalition formation is also related to other computational problems.
Most prominently is the matching problem. 
This can be interpreted as computing a partition where the size of coalitions is bounded by~$2$. 
In the coalition formation literature, bounded coalition sizes have been well studied \citep{Ball04a,PeEl15a,BrBu20a}, and hardness results can be observed once the size of coalitions exceeds~$2$.
By contrast, maximum weight matchings or stable matchings can be computed efficiently \citep{Edmo65a,GaSh62a,Irvi85a}.
Another related problem with a more coalitional flavor is correlation clustering, a significant problem in machine learning~\citep{BBC04a,Swam04a,DEFI06a}.
The input is usually a complete graph with edges labeled as ``$+$'' or ``$-$''.
These denote the similarity or dissimilarity of vertex pairs, respectively.
Additionally, edges can have a weight.
Hence, the input is very similar to that of ASHGs.
The goal is to find a partition that maximizes the total weight of ``$+$'' edges (agreements) inside clusters plus ``$-$'' edges across different clusters (disagreements). 
This is similar to social welfare maximization, but the utilities between coalitions also matter for the objective value.
Interestingly, this leads to better approximation guarantees than in ASHGs:
For instance, when minimizing disagreements, there exist polynomial-time algorithms achieving a constant-factor approximation for unweighted games (i.e., weights in the range $\{-1,1\}$) \citep{BBC04a} and an $\orderof{\log n}$-approximation 
for weighted games \citep{DEFI06a}.

All of the literature discussed thus far considers ASHGs and other models in an offline setting.
Online hedonic games came under scrutiny more recently and have been researched far less.
\citet{FMM+21a} introduce the problem and focus on deterministic algorithms with guarantees to social welfare.
They consider a class of games equivalent to ASHGs, with the difference that utilities are scaled by a factor of $2$.
In addition, they consider games restricted by a maximum coalition size or number of coalitions, as well as a preference encoding related to FHGs. 
For the latter, the greedy algorithm is once again optimal, achieving a similar competitive ratio, dependent on utilities.\footnote{Since the utilities are based on taking average utilities, the competitive ratio does, however, not depend on the number of agents.}
In follow-up work to the present paper, \citet{BRS25a} study online FHGs under random arrival and free dissolution.
\publ{replace link to fhg paper to link of permanent arxiv version?}
They establish further connections with matching algorithms and show that these algorithms lead to optimal or close to optimal competitive ratios.
In addition, \citet{CoAg25a} study ASHGs under the objective of maximizing egalitarian welfare, i.e., the minimum utility of an agent, which turns out to be a demanding objective.
Moreover, \citet{BuRo24a} consider online ASHGs with the goal of obtaining stable partitions.
They also consider Pareto optimality, which we have already discussed above as a natural weakening of maximizing social welfare:
If no pair of agents has a utility of~$0$, then Pareto optimality can be achieved online.
Otherwise, this is impossible, even if utilities are restricted to be symmetric and in $\{-1,0,1\}$. 

Somewhat similar to online models of coalition formation, there is a recent stream of research on dynamic models of coalition formation \citep{BFFMM18a,BBW21b,CMM19a}.
Instead of agents arriving over time, there exists an initial partition, which is altered over discrete time steps based on deviations of agents.
Since deviations are usually related to stability questions, the central question of this line of research is whether dynamics converge to a stable partition.
Some of this literature explicitly deals with ASHGs or close variants \citep{BMM22a,BBK25a,BBT23a,BuSu24a}.
 
\subsection{Online Algorithms}
We now review the literature on online algorithms.
Most related to our work is the online matching problem, which was first studied in the seminal paper by \citet{KVV90a}.
In this work, an unweighted, bipartite graph is given, and the agents of one side appear online while the agents of the other side are present from the beginning.
The goal is to find a matching of maximum cardinality.
\citet{KVV90a} introduce the famous ranking algorithm, which achieves a competitive ratio of $1-\frac{1}{e}$.
An overview of this research is presented in the book chapter by \citet{HuTr22a}.

Our model of online coalition formation can be viewed as a generalization of the setting by \citet{KVV90a}, with the following modifications: (i) there are edge weights, (ii) all vertices arrive online, (iii) the underlying graph is not necessarily bipartite, and (iv) coalitions can be arbitrary subsets of agents.
While condition (iv) is specific to coalition formation, conditions (i)-(iii) have been studied in the literature. 
Our focus is on the conjunction of all four conditions, but we will also identify connections with a general model of online matching where conditions (i)-(iii) are satisfied.

\citet{FKM+09a} consider condition (i), i.e., a bipartite setting with edge weights, where one side of the vertices arrives online.
They show that the natural greedy algorithm is $\frac{1}{2}$-competitive, and provide an algorithm matching the competitive ratio of $1-\frac{1}{e}$ by \citet{KVV90a}.
An important condition in this setting is \emph{free disposal}, which essentially requires that a previous matching can be dissolved upon the arrival of a better option.
\citet{WaWo15a} consider condition (ii), i.e., online bipartite matching where both sides arrive online.
For the fractional relaxation of the problem, they can beat the greedy algorithm by a primal-dual algorithm, achieving a competitive ratio of $0.526 < 1 - \frac{1}{e}$.
\citet{GKM+19a} consider the conjunction of conditions (ii) and (iii), i.e.,  online (non-bipartite) matching.
They show that, in this model, a competitive ratio better than $\frac 12$ is impossible to achieve by a deterministic algorithm.
Note that in contrast to \citet{WaWo15a}, they study the original integral problem.
However, they provide a randomized algorithm with a competitive ratio of $\frac 12 + \epsilon$ for a small constant $\epsilon > 0$.
Similarly, \citet{HKT+18a} also consider the conjunction of conditions (ii) and (iii), 
but they modify the setting so that vertices do not have to be matched immediately upon arrival.
Instead, vertices can be matched up until a particular deadline.
They extend the ranking algorithm to this setting, showing that it is $0.5211$-competitive.
Both \citet{WaWo15a} and \citet{HKT+18a} show that a competitive ratio of $1 - \frac{1}{e}$ is impossible to achieve in their respective settings.

The random arrival model has been studied extensively in online algorithms. 
Its most prominent application is in the secretary problem \citep{Ferg89a}, which can be seen as bipartite matching with a single offline vertex. 
The task is to maximize the probability of stopping with the most valuable edge.
For this problem, an optimal stopping algorithm achieves the optimal competitive ratio of $\frac 1e$ \citep{Dynk63a}.
This remains the best possible competitive ratio in the weighted bipartite matching setting where one side of the vertices is present offline \citep{KRTV13a}.
Optimal stopping algorithms have been studied in more generality \citep{Brus00a}, and we will study them in the context of coalition formation in \Cref{ap:odds}.
Interestingly, the competitive ratio of $1-\frac{1}{e}$ in the unweighted, bipartite setting according to \citet{KVV90a} can be beaten in the random arrival model \citep{KMT11a,MaYa11a}.
Closest to us is the work by \citet{EFGT22a} who study weighted, non-bipartite instances in which all vertices or edges arrive online.
For vertex arrival, they show a tight competitive ratio of $\frac 5{12}$.
However, this algorithm needs knowledge of the number of agents.

A related model that is 
a special case of the so-called semi-streaming model considers the arrival of edges instead of vertices \citep{FKM+05a}.
This setting can be strictly harder for unweighted, non-bipartite graphs than the one with vertex arrivals.
Indeed, a competitive ratio of $\frac 12 + \epsilon$ cannot be achieved by a randomized algorithm for any  $\epsilon > 0$ \citep{GKM+19a}.
Moreover, for the case of weighted instances, a finite competitive ratio is impossible to achieve if all edges arrive online, and this holds even for bipartite instances and randomized algorithms \citep{ELMS11a}.
Hence, a condition like the free disposal by \citet{FKM+09a} is necessary.
This condition is usually called \emph{preemption} in the edge arrival model.
A variant of the greedy algorithm, where an edge or a pair of edges is only disposed if this leads to an improvement of the weight of the disposed edges by a factor of $2$, is $\frac 16$-competitive \citep{FKM+05a}.
By optimizing the threshold for edge disposal, one can improve this to obtain a $\frac 1 {3+2\sqrt{2}}$-competitive ratio, where $3+2\sqrt{2}\approx 5.83$ \citep{McGr05a}.
Interestingly, this algorithm achieves the optimal competitive ratio for this setting \citep{Vara11a}.
In contrast to the unweighted setting, where edge arrival is strictly harder for randomized algorithms, we will show how to extend the techniques of \citet{McGr05a} and \citet{Vara11a} to achieve the optimal competitive ratio of $\frac 1 {3+2\sqrt{2}}$ in the case of vertex arrivals.
We want to remark that the streaming literature is motivated by the task of processing large data sets and, therefore, the goal of storing the current solution in small, often sublinear, space is usually considered in addition to obtaining a matching of a high weight \citep{Muth05a}.
In the context of matching, a space of $\tilde{O}(n)$, where $n$ is the number of vertices, is necessary.
Otherwise, even simpler properties like connectivity cannot be tested \citep{FKM+05b}.
Streaming under this space assumption is often referred to as semi-streaming.
In all of our algorithms, we only need to store the current partition and each coalition's value, which can be done in $\tilde{O}(n)$ space.

Finally, we want to discuss research on online versions of problems with a coalitional flavor.
These models typically have in common with coalition formation that the goal is to output a partition of agents (or objects), but they differ with regard to the utility structure.
First, \citet{CLMP22a} study the online variant of correlation clustering, as discussed before.
More loosely related is online skill formation, a coalitional variant of which was studied by \citet{CoAg23a}.
In this problem, agents have to be partitioned into groups aiming at the capability to perform specific tasks.
The value of a coalition depends on whether agents fulfill the skill requirements as a coalition.
Therefore, the main difference to our model is that the utility structure is multi-dimensional.

\section{Preliminaries}

In this section, we present our model.
For an integer $i\in \mathbb N$, we define $[i] := \{1,\dots,i\}$.
Also, for any set $N$, we define ${N \choose 2} := \{e \subseteq N \colon |e| = 2\}$.

\subsection{Additively Separable Hedonic Games}

Let $N$ be a finite set of $n$ \emph{agents}. 
Any subset of $N$ is called a \emph{coalition}. 
We denote the set of all possible coalitions containing agent $i\in N$ by $\mathcal N_i$, i.e., $\mathcal N_i := \{C \subseteq N \colon i \in C\}$.
A \emph{coalition structure} (or \emph{partition}) is a partition of the agents.
Given an agent $i \in N$ and a partition $\pi$, let $\pi(i)$ denote the coalition of $i$, i.e., the unique coalition $C \in \pi$ with $i \in C$. 

A \emph{hedonic game} is a pair $(N,\succsim)$ consisting of a set $N$ of agents and a preference profile ${\succsim} = (\succsim_i)_{i\in N}$, where $\succsim_i$ is a weak order over $\mathcal N_i$ that represents the preferences of agent~$i$. 
Following \citet{BoJa02a}, a hedonic game is called an \emph{additively separable hedonic game} (ASHG) if there exists a complete, undirected,\footnote{Since our focus will be on social welfare, the consideration of \emph{undirected} graphs is without loss of generality because the social welfare of a partition is invariant under the symmetrization $w^S(\{x,y\}) = \frac 12(w_x(y) + w_y(x))$ given directed edges with weights $w_i(j)$.
Note that this symmetrization is not without loss of generality for settings with a restricted utility range. 
For instance, when restricting to $\{-1,1\}$, the symmetrization can lead to a utility of~$0$.} 
and weighted graph $G = (N,E,w)$ with edge set $E = {N\choose 2}$ and weight function $w\colon E\to \mathbb Q$, such that, for every agent $i\in N$ 
and every pair of coalitions $C,C'\in \mathcal N_i$, it holds that $C\succsim_i C'$ if and only if 
$$\sum_{j\in C\setminus\{i\}}w(\{i,j\}) \ge \sum_{j\in C'\setminus\{i\}}w(\{i,j\})\text.$$

We then speak of the ASHG \emph{given by} $G$.
We abbreviate $w(i,j) := w(\{i,j\})$.
In addition, we extend the weight function to sets of edges $F\subseteq {N\choose 2}$ by $w(F) := \sum_{e\in F}w(e)$.
Moreover, since we only consider complete graphs, we shorten the notation and write $G = (N,w)$, where $w\colon {N\choose 2}\to \mathbb Q$, instead of $G = (N, E, w)$, to fully specify an underlying graph.
For an agent $i\in N$ and a coalition $C\in \mathcal N_i$ or a partition $\pi$, we define the \emph{utility} of $i$ for~$C$ or~$\pi$ by $u_i(C) := \sum_{j\in C \setminus \{i\}}w(i,j)$ and $u_i(\pi) := u_i(\pi(i))$, respectively.

We now define our primary objective. 
Given a partition $\pi$, its \emph{(utilitarian) social welfare} $\SW$ is defined by $$\SW = \sum_{i\in N} u_i(\pi)\text,$$
i.e., as the sum of all agents utilities.
Similarly, the social welfare of a coalition $C\subseteq N$ is the sum of the utilities for~$C$ of the agents in $C$, defined by $\SW[C] := \sum_{i\in C} u_i(C)$.
Note that $\SW = \sum_{C\in \pi}\SW[C]$.
Our goal is to compute partitions of high social welfare.
A partition $\pi$ is said to be \emph{welfare-optimal} or an \emph{optimal partition} if, for every partition~$\pi'$, it holds that $\SW\ge\SW[\pi']$.
Given a hedonic game $G$, we usually denote optimal partitions by $\pi^*(G)$, omitting the game when it is clear from the context.

A \emph{matching} is a coalition structure $\pi$ such that, for all $C\in \pi$, it holds that $|C|\le 2$.
For better readability, we use the letter $\genmat$ instead of $\pi$ when referring to a matching.
A matching $\genmat$ is represented by its edge set $M(\genmat) := \{C\in \genmat\colon |C| = 2\}\subseteq {N\choose 2}$.
We then write $w(\genmat) := w(M(\genmat))$ for the weight of matching~$\genmat$.
Note that $\SW[\genmat] = 2 w(\genmat)$ as every edge contributes to the utility of precisely the agents representing its endpoints. 
Hence, maximizing social welfare among matchings is identical to finding a maximum weight matching.

\subsection{Online Coalition Formation}

In this section, we introduce our model of online coalition formation and appropriate objectives.
We remark that we only consider \emph{deterministic} algorithms, while randomization only occurs in the random arrival model.

Consider an ASHG given by $G = (N,w)$.
Given a subset of agents $N'\subseteq N$, let $G[N']$ denote the subgraph induced by agent set $N'$.
Moreover, given a partition $\pi$ of~$N$ and a subset of agents $N'\subseteq N$, we define $\pi[N']$ as the \emph{partition restricted to $N'$}, i.e., $\pi[N'] := \{C\cap N'\colon C\in \pi, C\cap N'\neq \emptyset\}$.
Specifically, if $N' = N\setminus \{i\}$ for some agent $i\in N$, we write $\pi - i$ instead of $\pi[N']$.

In an online setting, previous decisions influence the capabilities of an algorithm to form a partition in the next step.
More precisely, an algorithm iteratively builds a partition such that whenever an agent arrives, the only knowledge is the game restricted to the present agents, and the algorithm has to irrevocably assign the new agent to an existing coalition or start a new coalition (or, under free dissolution, dissolve any coalition before making its decision).

Given a partition $\pi$ and an agent $i$ not contained in a coalition of $\pi$, let $\avai(\pi,i)$ denote the set of \emph{available partitions}, 
when the tentative (partial) partition is $\pi$ and the newly arriving agent is $i$.
As a default, we assume the standard setting where $\avai(\pi,i) = \avai^S(\pi,i) := \{\pi'\colon \pi' - i = \pi\}$.
We also consider algorithms that can dissolve a coalition completely.
We say that an algorithm acts under \emph{free dissolution} if $\avai(\pi,i) = \avai^D(\pi,i) := \avai^S(\pi,i)\cup \bigcup_{C\in \pi, j\in C}\{(\pi\setminus \{C\})\cup \{\{i,j\}\}\cup\{\{k\}\colon k\in C\setminus\{j\}\}\}$. 
In other words, the algorithm may take any existing coalition, create a coalition of size~$2$ with the new agent and some agent in this coalition, and put all other agents into singleton coalitions.
Free dissolution is a natural extension of free disposal in the domain of matching by \citet{FKM+09a} adapted to coalitions of size larger than~$2$.
In addition, we define $\Sigma(N) := \{\sigma\colon [|N|] \to N \textnormal{ bijective}\}$ as the set of all \emph{orders} of the agent set $N$.

An instance of an \emph{online coalition formation problem} is a tuple $(G,\sigma)$, where $G = (N,w)$ defines an ASHG and $\sigma\in \Sigma(N)$.
An \emph{online coalition formation algorithm} for instance $(G,\sigma)$ gets as input the sequence $G_1,\dots,G_n$, where, for every $i\in [n]$, $G_i = G[\{\sigma(1),\dots,\sigma(i)\}]$.
Then, for every $i\in [n]$, the algorithm produces a partition $\pi_i$ of $\{\sigma(1),\dots,\sigma(i)\}$ such that 
the algorithm has only access to $G_i$ and
for $i\ge 2$, it holds that $\pi_i\in \avai(\pi_{i-1},\sigma(i))$.
The output of the algorithm is the partition $\pi_n$.
Given an online coalition formation algorithm $\alg$, let $\alg(G,\sigma)$ be its output for instance $(G,\sigma)$.

If an online coalition formation algorithm creates a matching in every step, we speak of an \emph{online matching algorithm}.
Note that we chose an unusual way of introducing matching algorithms as special types of coalition formation algorithms. 
However, this accounts for matchings being special partitions, and, therefore, the matching  setting being a special case of our coalition formation setting.

Our benchmark algorithm is the greedy algorithm as introduced by \citet{FMM+21a}.

\begin{definition}[Greedy algorithm]
	On input $(G,\sigma)$, in the $i$th step, $i\ge 2$, the \emph{greedy algorithm} ($\gdy$) forms $\pi_i = \arg\max_{\pi\in \avai(\pi_{i-1},\sigma(i))}\SW[\pi]$ if there exists $\pi\in \avai(\pi_{i-1},\sigma(i))$ with $\SW[\pi] > \SW[\pi_{i-1}]$, and it forms $\pi_i = \pi_{i-1}\cup\{\{\sigma(i)\}\}$, otherwise.
\end{definition}
Hence, $\gdy$ assigns each arriving agent to the available coalition that maximizes the increase in social welfare or creates a new singleton coalition if no increase is possible.

\subsection{Competitive Analysis}\label{sec:solcon}

Our goal is to maximize the social welfare of the partition produced by an online algorithm.
The performance of an algorithm is measured by its competitive ratio, which captures the worst-case approximation to the maximum social welfare of its returned partitions.
Recall that $\pi^*(G)$ denotes an optimal partition in game $G$.
We say that an online coalition formation algorithm $\alg$ is \emph{$c$-competitive}\footnote{Here, we use the convention that $\frac 00 = 1$ and $\frac x0 = 0$ for any $x\in \mathbb Q$ with $x < 0$.
Also, note that \citet{FMM+21a} defines the competitive ratio inversely so that it is always at least $1$.
Here, we prefer the more common definition in the online matching literature.} if
\[
\inf_G\min_{\sigma\in \Sigma(N)}\frac{\SW[\alg(G,\sigma)]}{\SW[\pi^*(G)]}\ge c\text.
\]

Equivalently, this means that, for all instances $(G,\sigma)$, it holds that
$\SW[\alg(G,\sigma)]\ge c\, \SW[\pi^*(G)]$.

In addition, we consider online coalition formation with a random arrival order, where we assume that the arrival order is selected uniformly at random.
In the random arrival model, an algorithm $\alg$ is said to be \emph{$c$-competitive} if 
\[
\inf_G\frac{\EV_{\sigma}\left[\SW[\alg(G,\sigma)]\right]}{\SW[\pi^*(G)]}\ge c\text.
\]

The expectation is over the uniform selection of an arrival order $\sigma$ from $\Sigma(N)$.
In both models, the \emph{competitive ratio} of an algorithm $\alg$, denoted by $c_\alg$, is the supremum $c$ such that $\alg$ is $c$-competitive.
In particular, the competitive ratio is always at most~$1$.

Note that if $\genmat$ is a matching, then $w(\genmat) = \frac 12 \SW[\genmat]$.
Hence, in the competitive analysis of online matching algorithms, we also consider the weight of matchings instead of their social welfare.

\section{Matching Algorithms as Coalition Formation Algorithms}

Recall that every matching algorithm can be used as a coalition formation algorithm as a matching is a special case of a partition.
In this section, we want to explore the guarantees achieved by matching algorithms when they are indeed used as coalition 
formation algorithms.
This is based on two insights captured in the next lemmas.
The first lemma is a ``folklore'' property about maximum weight matchings.
It follows as a consequence of Baranyai’s factorization theorem~\citep{Bara73a}.\footnote{We thank an anonymous ACM TALG reviewer for suggesting this simple and elegant proof.}

\begin{restatable}{lemma}{LemAvgmat}\label{lem:avgmat}
	Let $G = (N,w)$ be a complete weighted graph and $\genmat^*$ a maximum weight matching of $G$.
	Then, $w(\genmat^*)\ge \frac 1n w(E^+)$, where $E^+ = \left\{e\in {N\choose 2}\colon w(e)>0\right\}$.
\end{restatable}

\begin{proof}
	Let $G = (N,w)$ be an arbitrary complete and weighted graph and $\genmat^*$ a maximum weight matching of $G$.
	Let $E^+ = \left\{e\in {N\choose 2}\colon w(e)>0\right\}$ be the set of edges of positive weight.
	Since a perfect matching is a $1$-factor of a graph, it follows as a special case of Baranyai’s factorization theorem \citep{Bara73a} that every complete graph with an even number of vertices can be partitioned into $n-1$ perfect matchings.
	If we restrict these perfect matchings to edges of positive weight, then at least one of them has weight at least $\frac{1}{n - 1}w(E^+) \ge \frac{1}{n}w(E^+)$ because they partition the set of positive edges into $n-1$ sets.
	The argument easily extends to graphs with an odd number of vertices by adding a dummy vertex with zero weights to all other agents.
	Then, at least one perfect matching restricted to edges of positive weight has weight at least $\frac{1}{n}w(E^+)$.
	The claim $w(\genmat^*)\ge \frac 1n w(E^+)$ follows since the maximum weight matching is at least as large as the matching we found.
\end{proof}

Now, as every edge contributes to the utility of at most two agents, we can bound the weight of an optimal partition $\pi^*(G)$ in an ASHG $G$ as $\SW[\pi^*(G)] \le 2w(E^+)$.
Consequently, computing maximum weight matchings results in an $\frac 1n$-approximation for maximizing social welfare in the coalition formation domain.
This translates to the competitive ratio of online coalition formation algorithms as we capture in the next lemma.
We remark that the statement is true for both the random arrival and the free dissolution models.

\begin{lemma}\label{lem:mat2cofo}
	Let $\alg$ be a $c$-competitive algorithm for online matching.
	Then, $\alg$ is $\frac{c}{n}$-competitive for online coalition formation.
\end{lemma}
\begin{proof}
	Let $(G,\sigma)$ be an arbitrary instance of online coalition formation.
	Let $\genmat^*$ be a maximum weight matching for the underlying weighted graph $G$, and let $\pi^*$ be a partition maximizing social welfare.
	Let $E^+ = \{e\in E\colon w(e)>0\}$ be the set of positive edges.
	Then,
	\[
	\SW[\alg(G,\sigma)] \ge c\, 2\, w(\genmat^*) \ge c\, \frac 2n\, w(E^+) \ge \frac cn\, \SW[\pi^*]\text.
	\]
	
	The first inequality holds because $\alg$ is a $c$-competitive algorithm in the matching domain.
	The second inequality follows from \Cref{lem:avgmat}.
	The last inequality holds because twice the sum of positive edges is an upper bound for the social welfare of any partition. 
\end{proof}

Hence, online matching algorithms with a constant competitive ratio yield coalition formation algorithms with a linear competitive ratio.
While this seems not to be much, it is essentially the best we can ask for from a polynomial-time algorithm.
Indeed, recall from our discussion of related work that, for every $\epsilon > 0$, it is impossible to approximate maximum social welfare by a factor of $\frac 1 {n^{1-\epsilon}}$ in polynomial time, unless \P\,=\,\NP{} \citep{FKV22a}.

Hence, some of our effort will be transferring matching algorithms while dealing with additional obstacles: under random arrival, we will have to make algorithms work for an unknown number of agents and for matching under free disposal, we will transfer results from the edge arrival to the vertex arrival setting.

\section{Random Arrival Model}\label{sec:RAM}

In this section, we analyze the competitive ratio of algorithms in the random arrival model.
For our analysis of algorithms in this section, we use the notation $x\arr y$ to say that $\sigma^{-1}(x) < \sigma^{-1}(y)$ for $x,y\in N$ and an arrival order $\sigma$.

\subsection{Algorithmic Possibilities}
We start by analyzing the greedy algorithm, which is known to achieve an optimal competitive ratio for deterministic arrivals \citep{FMM+21a}.
Recall that, in the deterministic arrival model, the specific utility values significantly affect the competitive ratio of $\gdy$ as well as the one of any other online algorithm \citep{FMM+21a}.
By contrast, we now show that the competitive ratio of $\gdy$ in the random arrival model solely depends on the number of agents.

\begin{theorem}\label{thm:compGDY}
	The competitive ratio of $\gdy$ for online coalition formation under random arrival order satisfies $\Theta\left(\frac{1}{n^2}\right)$.
\end{theorem}

\begin{proof}
	\begin{figure}[tbp]
		\begin{tikzpicture}[-, node distance=2.5cm,main/.style={draw,}]
			\node[draw, circle](a) at (2,3){$a$};
			\node[draw, circle](b) at (4,3){$b$};
			\path
			(a) edge node[midway, fill = white]{$1$} (b);
			\draw (-.5,3) ellipse (10pt and 20pt) node[above=2em]{$X$};
			\node(x) at (-.2, 3) {};
			\path (a) edge node[pos = 0.6, fill = white] (au) {$\epsilon$} ($(x)+(-.03,.6)$);
			\path (a) edge node[pos = 0.6, fill = white] (ad) {$\epsilon$} ($(x)+(-.03,-.6)$);
			\node[rotate = 90] at ($(au)!.5!(ad)$) {\small \dots};
			\draw (6.5,3) ellipse (10pt and 20pt) node[above=2em]{$Y$};
			\node(y) at (6.2, 3) {};
			\path (b) edge node[pos = 0.6, fill = white] (bu) {$\epsilon$} ($(y)+(.03,.6)$);
			\path (b) edge node[pos = 0.6, fill = white] (bd) {$\epsilon$} ($(y)+(.03,-.6)$);
			\node[rotate = 90] at ($(bu)!.5!(bd)$) {\small \dots};
		\end{tikzpicture}
		\centering
		\caption{Hard instance for online coalition formation algorithms under random arrival. The example contains $n = 2k + 2$ agents. There are $k$ agents in each of the sets $X$ and $Y$. 
			The utility between $a$ and any agent in $X$ and $b$ and any agent in $Y$ is $\epsilon$. 
			All omitted edges represent weights of $-1$. 
		}
		\label{fig:hard_instance}
	\end{figure}
	
	First, we show that the competitive ratio of $\gdy$ satisfies $\orderof{\frac{1}{n^2}}$ by providing a family of instances where it performs poorly. 
	Let $\epsilon > 0$ and consider the ASHG given by $G^{k,\epsilon} = (N^{k,\epsilon},w^{k,\epsilon})$ depicted in \Cref{fig:hard_instance}.
	Since, we analyse this graph for a fixed $k$, we omit this superscript from now on and write $G^{\epsilon} = (N^{\epsilon},w^{\epsilon})$.
	
	The agent set is $N^{\epsilon} = X\cup Y\cup \{a,b\}$, where $|X| = |Y| = k$. 
	Weights are given by $w^\epsilon(a,b) = 1$, $w^\epsilon(a,x) = w^\epsilon(b,y) = \epsilon$ for $x\in X$ and $y\in Y$, and all other weights are $-1$.
	For sufficiently small $\epsilon$, the optimal partition has a value of $2$ (with $\{a,b\}$ the only non-singleton coalition). 
	By inspecting the limit case for $\epsilon$ tending to $0$, the value of the greedy algorithm (and, therefore, also its competitive ratio) is at most twice the probability of forming $\{a,b\}$.
	Let $\epsilon^* = \frac 12$.
	Then it holds that
	\begin{small}
	\begin{align*}
		c_\gdy & = \inf_{G}\frac{\EV_{\sigma}\left[\SW[\gdy(G,\sigma)]\right]}{\SW[\pi^*(G)]}
		\le \inf_{\epsilon > 0}\frac{\EV_{\sigma}\left[\SW[\gdy(G^\epsilon,\sigma)]\right]}{2}\\
		&\le  \inf_{\epsilon > 0}\pr_\sigma(\{a,b\}\in\gdy(G^\epsilon,\sigma)) + \sum_{x\in X}\epsilon \pr_\sigma(\{x,a\}\in\gdy(G^\epsilon,\sigma)) + \sum_{y\in Y}\epsilon \pr_\sigma(\{y,b\}\in\gdy(G^\epsilon,\sigma))\\
		&\le \inf_{\epsilon > 0}\pr_\sigma(\{a,b\}\in\gdy(G^\epsilon,\sigma)) + 2\epsilon 
		= \inf_{\epsilon > 0}\pr_\sigma(\{a,b\}\in\gdy(G^{\epsilon^*},\sigma)) + 2\epsilon 
		\stackrel{\epsilon \to 0}{=} \pr_\sigma(\{a,b\}\in\gdy(G^{\epsilon^*},\sigma))\text.
	\end{align*}
	\end{small}
	
	In the second inequality, we use that, for small values of $\epsilon$, $\gdy$ always either forms the partition $\{\{a,b\}\} \cup \{\{z\}\colon z \in X \cup Y\}$ or $\{\{x^*, a\}, \{y^*,b\}\} \cup \{\{z\}\colon z \in X \cup Y \setminus \{x^*, y^*\}\}$ for $x^* \in X$ and $y^* \in Y$.
	The third inequality follows because the probabilities in the sums concern disjoint events, i.e., $\sum_{x\in X}\pr_\sigma(\{x,a\}\in\gdy(G^\epsilon,\sigma))  \le 1$. 
	In the last row, we use that the execution of the greedy algorithm is identical for all $G^\epsilon$ with $0 < \epsilon < 1$.
	Thus, while the infimum is achieved for $\epsilon$ tending to $0$, we have $\pr_\sigma(\{a,b\}\in\gdy(G^\epsilon,\sigma)) =  \pr_\sigma(\{a,b\}\in\gdy(G^{\epsilon^*},\sigma))$ for all $0 < \epsilon < 1$.
	
	We further compute
	\begin{align*}
		&\pr_\sigma(\{a,b\}\in\gdy(G^{\epsilon^*},\sigma))\\
		& = \ \pr_\sigma(\{a,b\}\in\gdy(G^{\epsilon^*},\sigma)\mid a \arr b)\pr_\sigma(a \arr b)
		 + \pr_\sigma(\{a,b\}\in\gdy(G^{\epsilon^*},\sigma)\mid b \arr a)\pr_\sigma(b \arr a)\\
		&= \ 2 \pr_\sigma(\{a,b\}\in\gdy(G^{\epsilon^*},\sigma)\mid a \arr b)\pr_\sigma(a \arr b)\text.
	\end{align*}
	
	The second equality follows by symmetry.
	
	The next step is to sum over all possible arrival positions of $b$ by summing over the number of alternatives that,  in addition to $a$, arrive before $b$. 
	Note that if more than $k$ agents arrive before $b$, excluding $a$, then, by the pigeonhole principle, some $x \in X$ arrives before $b$ and forms a coalition with $a$. 
	This prevents the coalition $\{a, b\}$ from forming, so all terms of the resulting sum for $i > k$ are $0$. 
	Conditioned on $a$ arriving before $b$, the coalition $\{a, b\}$ forms if and only if all $i$ agents arriving before $b$ are from $Y$ since those agents will not form a coalition with $a$. 
	We derive\footnote{For the first inequality, note that it holds that $\pr(A\mid B)\pr(B) = \pr(A,B) = \sum_C\pr(A,B,C) = \sum_C \pr(A\mid B,C)\pr(B,C) = \sum_C \pr(A\mid B,C)\pr(B\mid C)\pr(C)$ for arbitrary events $A$, $B$, and $C$ and probability measures $\pr$.} 
\begin{small}
	\begin{align*}
		c_\gdy &\le 2 \sum_{i = 0}^{k} \pr_\sigma(\{a,b\}\in\gdy(G^{\epsilon^*},\sigma)\mid a \arr b, \sigma^{-1}(b) = i + 2) 
\cdot \pr_\sigma(a\arr b\mid \sigma^{-1}(b) = i + 2)\cdot \pr_\sigma(\sigma^{-1}(b) = i + 2) \\
		&= 2 \sum_{i = 0}^{k} \pr_\sigma(\{d\colon d \arr b\} \subseteq Y \cup \{a\} \mid a \arr b, \sigma^{-1}(b) = i + 2) \cdot \pr_\sigma(a\arr b\mid \sigma^{-1}(b) = i + 2)\cdot \pr_\sigma(\sigma^{-1}(b) = i + 2) \\
		&= 2 \sum_{i = 0}^{k} \pr_\sigma(\{d\colon d \arr b\} \setminus \{a\} \subseteq Y \mid a \arr b, \sigma^{-1}(b) = i + 2) 
		\cdot \pr_\sigma(a\arr b\mid \sigma^{-1}(b) = i + 2)\cdot \pr_\sigma(\sigma^{-1}(b) = i + 2)\text.
	\end{align*}
\end{small}
	
	We now determine the value of each probability in the previous sum.
	First, the probability that the $i$ agents arriving before $b$ are all from $Y$ is $\pr_\sigma(\{d\colon d \arr b\} \setminus \{a\} \subseteq Y \mid a \arr b, \sigma^{-1}(b) = i + 2) = \frac{\binom{k}{i}}{\binom{2k}{i}}$, i.e., the number of possibilities to draw $i$ agents from $Y$ divided by the number of possibilities to draw $i$ agents from $Y \cup X$. 
	
	Second, the probability that $a$ arrives before $b$ when $b$ arrives in position $i + 2$ is $\pr_\sigma(a \arr b\mid \sigma^{-1}(b) = i + 2) = \frac{i + 1}{2k + 1}$ because we have $i + 1$ chances to draw $a$ among the remaining $2k + 1$ alternatives.
	Finally, due to the random arrival order, the probability that agent $b$ arrives in a particular fixed position is $\pr_\sigma(\sigma^{-1}(b) = i + 2) = \frac{1}{2k + 2}$. 
	Together, we obtain 
	\begin{equation*}
		c_\gdy \le 2 \sum_{i = 0}^{k} \frac{\binom{k}{i}}{\binom{2k}{i}} \frac{i + 1}{2k + 1}\frac{1}{2k + 2} = \frac{2}{k^2 + 3k + 2} \in \orderof{\frac{1}{k^2}} =  \orderof{\frac{1}{n^2}}\text.
	\end{equation*}
	
	Next, we show that the competitive ratio of $\gdy$ satisfies $\Omega\left(\frac{1}{n^2}\right)$. 
	Consider an arbitrary ASHG given by $G = (N,w)$ and let $E^+(G) := \left\{e\in {N \choose 2}\colon w(e) > 0\right\}$ be the set of agent pairs with positive weights.
	Then, the optimal partition $\pi^*(G)$ satisfies $\SW[\pi^*(G)] \le 2\sum_{e \in E^+(G)}w(e)$. 
	Furthermore, the social welfare of $\gdy$ is at least twice the utility between the first arriving pair of agents from $E^+(G)$.
	Indeed, until a pair of agents with positive utility arrives, every agent is assigned to a singleton coalition. 
	Thus, since every pair in $E^+(G)$ has equal probability of being the first such pair, the expected social welfare of $\gdy$ is at least the average, i.e.,  $\EV_\sigma[\SW[\gdy(G,\sigma)]] \ge \frac{\sum_{e \in E^+(G)}2 w(e)}{|E^+(G)|}$ and the competitive ratio is then at least
	\begin{align*}
		c_\gdy &= \inf_{G}\frac{\EV_\sigma[\SW[\gdy(G,\sigma)]]}{\SW[\pi^*(G)] }\ge \inf_{G}\frac{\frac{\sum_{e \in E^+(G)}2w(e)}{|E^+(G)|}}{2\sum_{e \in E^+(G)}w(e)} = \inf_{G}\frac{1}{|E^+(G)|} \in \Omega\left(\frac{1}{n^2}\right)\text.
	\end{align*}
	
	Altogether, we have shown that $c_\gdy$ is of order $\Theta\left(\frac{1}{n^2}\right)$.
\end{proof}

The natural question is whether we can obtain better algorithms than the greedy algorithm.
As we have just seen, the asymptotic performance guarantee of the greedy algorithm is obtained by considering the average weight of a positive edge.
Hence, its competitive ratio is asymptotically as bad as the competitive ratio of the algorithm that forms a single coalition of size~$2$ with the first pair of agents arriving with a positive utility, and forming singleton coalitions with all other agents.
However, we can easily achieve the average weight of a random matching, improving the performance to $\Theta\left(\frac 1n\right)$.
For simplicity, we assume first that $n$ is even and known to the algorithm in advance.
We first analyze a simple online matching algorithm.\footnote{
	An alternative way to achieve a competitive ratio of $\Theta\left(\frac 1n\right)$ for online coalition formation under random arrival with known~$n$ is to use the \emph{randomized} online matching algorithm by \citet{EFGT22a} with a constant competitive ratio, and to apply \Cref{lem:mat2cofo}.
By contrast, we propose here a simple \emph{deterministic} algorithm that only slightly adjusts the greedy algorithm and can form coalitions of arbitrary size. 
Its performance in the coalition formation domain is asymptotically as good as in the matching domain.
The subsequent step of transforming the obtained algorithm into an algorithm for unknown $n$ is necessary in both approaches.}

\begin{definition}[Greedy matching algorithm]
	The \emph{greedy matching algorithm} ($\rmat$) leaves the first $\frac n2$ agents unmatched. 
	Then, for $1\le i\le \frac n2$, if the weight between the $\left(\frac n2 + i\right)$th agent and the $i$th agent is positive, these are matched.
	Otherwise, the $\left(\frac n2 + i\right)$th agent remains unmatched.
\end{definition}

We now determine the competitive ratio of $\rmat$ in the matching domain.
Afterward, we consider a slightly more sophisticated version of this algorithm and prove the same performance guarantee in the coalition formation domain.

\begin{restatable}{proposition}{propRMAT}\label{prop:propRMAT}
	$\rmat$ has a competitive ratio of $\Theta\left(\frac{1}{n}\right)$ for online matching under random arrival when $n$ is known and even.
\end{restatable}

\begin{proof}
	We first show an upper bound for the performance of $\rmat$.
	Consider an instance given by $G = (N,w)$ such that there exist agents $a,b\in N$ with $w(a,b) = 1$, and $w(x,y) = 0$ if $\{x,y\}\neq \{a,b\}$. 
	Then the maximum weight matching has weight $1$, but the probability that $\rmat$ matches $a$ with $b$ is 
	\begin{equation}\label{eq:RMAcomp}
		\frac{\binom{2}{1}\binom{n-2}{n/2-1}}{\binom{n}{n/2}}\cdot \frac 2n = \frac n {2(n-1)}\cdot \frac 2n = \frac 1 {n-1} = \Theta\left(\frac 1n\right) \text,
	\end{equation}
	
	The first part of the product is the probability that $a$ and $b$ appear in different stages of the algorithm, and the factor of $\frac{2}{n}$ is the probability that they are matched conditioned on them arriving in different stages.
	For the latter, note that for a fixed agent in the second stage, an edge with each of the $\frac n2$ agents of the first stage is proposed with equal probability.
	Hence, the competitive ratio of $\rmat$ satisfies $\orderof{\frac 1n}$.
	
	We proceed with the proof of the lower bound of the performance guarantee.
	As before, in an execution of $\rmat$, every positive edge has a probability of $\frac 1{n-1}$ to contribute to the computed matching.
	Let $E^+ = \left\{e\in {N \choose 2}\colon w(e) > 0\right\}$. 
	Hence, $\EV_\sigma[\rmat(G,\sigma)] = \sum_{e\in E^+} \frac 1 {n-1} w(e) \ge \frac 1{n-1} w(E^+)$.
	Since any matching is of weight at most $w(E^+)$, we conclude that
	\begin{equation*}
		c_\rmat \ge \frac 1{n-1} = \Omega\left(\frac 1n\right)\text.
	\end{equation*}	

	This concludes the proof.
\end{proof}

Notably, $\rmat$ is not sensitive to edge weights:
When an agent arrives in the second stage of the algorithm, it might have a very valuable option that it is not allowed to match with.
Hence, we now define a refined variant where the second stage is performed greedily, which we analyze in the coalition formation domain.
However, it is unclear why this would improve the algorithm's performance.
Indeed, an earlier greedy decision might lead to missing out on even more welfare later on.
We show that this effect does not lead to an overall welfare loss, and the obtained algorithm is at least as good as $\rmat$. 
It is natural to ask whether even the worst-case performance has improved.
Unfortunately, this is not the case.
While the performance is in fact better in the simple worst-case example from the proof of \Cref{prop:propRMAT}, it still achieves a competitive ratio of $\orderof{\frac 1n}$ in the family of instances from \Cref{fig:hard_instance}.

First, we again assume that $n$ is even and known to the algorithm in advance. 
However, after we have analyzed this algorithm, we will show that we can drop this additional assumption and that a variation of the algorithm still achieves the same competitive ratio asymptotically.

\begin{definition}[Waiting greedy algorithm]
	The \emph{waiting greedy algorithm} ($\wgdy$) is defined for known and even $n$. 
	It places the first $\frac{n}{2}$ agents in singleton coalitions. 
	Then, it assigns coalitions greedily for the remaining $\frac{n}{2}$ agents.
\end{definition}

As we said, the upper bound of the performance of $\wgdy$ is attained by the games depicted in \Cref{fig:hard_instance}.
The analysis is more involved and relies on investigating the distribution of the agents in the second phase.
We defer the proof of the worst-case behavior of $\wgdy$ to \Cref{app:UppBounds}, and restrict attention to its algorithmic guarantee.

\begin{restatable}{theorem}{WGDY}\label{thm:wgdy}
	The competitive ratio of $\wgdy$ for online coalition formation under random arrival with known and even~$n$ is $\Theta\left(\frac{1}{n}\right)$.
\end{restatable}

\begin{proof}[Proof of lower bound]
	We show that $\wgdy \in \Omega\left(\frac{1}{n}\right)$. 
	Consider an arbitrary ASHG given by $G = (N,w)$.
	Let $\Pi(n) = \left\{(A,B)\colon A\cup B = N, |A| = |B| = \frac n2\right\}$ be the set of partitions of the agent set $N$ into two equally-sized subsets.
	
	Let $(A,B)\in \Pi(n)$. 
	Define $$E^+(A,B) = \left\{e = \{a,b\}\in {N\choose 2}\colon w(e) > 0, a\in A, b\in B\right\}\text.$$
	
	We claim that if the agents in $A$ and $B$ arrive in the first and second stage, respectively, 
	then the weight of the obtained partition is at least 
	$\frac 1n w(E^+(A,B))$.
	Let $J_{A,B}$ denote the event that the partition $(A,B)$ realizes.
	
	Consider an arbitrary agent $b\in B$.
	Let $a_1(b),\dots, a_{r(b)}(b)\in A$ be the agents in $A$ such that $\{b,a_i(b)\}\in E^+(A,B)$, where $r(b)\in \mathbb N$ is their number. 
	Assume that $w(a_1(b),b) \ge w(a_2(b),b) \ge \dots \ge w(a_{r(b)}(b),b)$.
	Let $i\in [r(b)]$. 
	In the event $J_{A,B}$, an agent $b\in B$ arrives as the $(\frac n2+i)$th agent (i.e., the $i$th agent within~$B$) with a probability of $\frac 2n$. 
	Moreover, when $b$ arrives as the $(\frac n2+i)$th agent, then at most $i-1$ coalitions have formed so far, and, therefore, some agent in $\{a_1(b),\dots, a_i(b)\}$ is still in a singleton coalition. 
	Hence, the increase in weight caused by $b$ is at least the weight of forming a coalition with the worst partner in $\{a_1(b),\dots, a_i(b)\}$, that is, $w(a_i(b),b)$.
	Let $X_b$ be the random variable denoting the gain in social welfare caused by the arrival of $b$.
	
	Hence, in the event $J_{A,B}$, the attained expected social welfare satisfies 
	\begin{align*}
		\EV_\sigma&[\SW[\wgdy(G,\sigma)]\mid J_{A,B}] \ge \sum_{b\in B} \EV_\sigma[X_b\mid J_{A,B}]\\ 
		&= \sum_{b\in B} \sum_{i = 1}^{r(b)} \EV_\sigma\left[X_b\mid J_{A,B}, \sigma^{-1}(b) = \frac n2 + i\right]\pr_\sigma\left(\sigma^{-1}(b) = \frac n2 + i\mid J_{A,B}\right)\\
		& = \sum_{b\in B} \sum_{i = 1}^{r(b)} \EV_\sigma\left[X_b\mid J_{A,B}, \sigma^{-1}(b) = \frac n2 + i\right] \frac 2n\\
		& \ge  \sum_{b\in B} \sum_{i = 1}^{r(b)} w(a_i(b),b) \frac 2n = \frac 2n \sum_{b\in B}\sum_{e\in E^+(A,B)\colon b\in e} w(e) =  \frac 2n w(E^+(A,B))\text.
	\end{align*}
	
	Hence, the expected social welfare of the partition computed by $\wgdy$ is
	\begin{equation*}
		\EV_{\sigma}\left[\SW[\wgdy(G,\sigma)]\right] \ge \frac 1{|\Pi(n)|}\sum_{(A,B)\in \Pi(n)}\frac 2n w\left(E^+(A,B)\right)\text. 
	\end{equation*}
	
	We use that each partition in $\Pi(n)$ is realized with equal probability.
	
	Now, there are a total of $|\Pi(n)| = {n \choose \frac n2}$ subdivisions into $A$ and $B$.
	Moreover, an arbitrary edge is contained in ${2 \choose 2}{{n-2}\choose {\frac n2 -2}}$ of these subdivisions.
	Hence, we obtain that
	\begin{equation*}
		\frac 1{|\Pi(n)|}\sum_{(A,B)\in \Pi(n)}w\left(E^+(A,B)\right) = \frac{\binom{2}{1}\binom{n-2}{n/2-1}}{\binom{n}{n/2}} w(E^+)= \frac n {2(n-1)}w(E^+)\ge \frac 12w(E^+)\text.
	\end{equation*}

	We conclude that
	\begin{equation}\label{eq:lbWGDY}
	\EV_{\sigma}\left[\SW[\wgdy(G,\sigma)]\right] \ge \frac 12\frac2n w\left(E^+\right)\text.
	\end{equation}

	Additionally, the social welfare of any partition is at most $2 w(E^+)$. 
	Combining this with \Cref{eq:lbWGDY}, we obtain
	$$
		c_\wgdy \ge \frac{\frac{w\left(E^+\right)}{n}}{2w(E^+)} = \frac{1}{2n} \in \Omega\left(\frac{1}{n}\right)\text.\eqno\qedhere
	$$
\end{proof}

Next, we show that we can still use a variant of this algorithm even if we do not know~$n$. 
The idea is to repeatedly run $\wgdy$ on an exponentially growing number of agents. 
This idea can be used to transform any algorithm that is only defined for a known and even number of agents.
Therefore, we introduce the algorithm in a more general framework.

\begin{definition}[Iterated doubling algorithms]
	Let $\alg$ be any online coalition formation algorithm for known and even $n$.
	The \emph{iterated doubling variant} of $\alg$ ($\ialg$) proceeds as follows:
	It maintains a parameter $i$ that is set to $i = 0$ initially and increased by $1$ whenever the next $2^{i+1}$ agents have arrived.
	We refer to the time during which the counter is set to a certain value $i$ as the $i$th \emph{phase}.
	In the $i$th phase, $\ialg$ applies $\alg$ to the agents arriving in the $i$th phase assuming that $2^{i+1}$ agents arrive.
	
	Specifically, we refer to $\mathit{I}$-$\wgdy$ as the \emph{iterated waiting algorithm} ($\iew$).
\end{definition}

	Hence, in the $i$th phase, $\iew$ places $2^{i}$ agents in singleton coalitions and then assigns $2^{i}$ agents to coalitions with the previous $2^{i}$ agents greedily.
	Note that $\ialg$ may not finish the last phase.
	In this case, we run $\alg$ on a smaller number of agents, assuming that a larger number of agents would arrive.
	The important insight is that the performance of $\ialg$ can only be worse than that of $\alg$ by a constant factor.

\begin{lemma}\label{lem:doubling}
	Assume that $\alg$ is a $c$-competitive algorithm for online coalition formation under random arrival with known and even $n$.
	Then, $\ialg$ is $\frac c{32}$-competitive for online coalition formation under random arrival (with unknown $n$).
\end{lemma}

\begin{proof}
	Assume that $\alg$ is a $c$-competitive algorithm for online coalition formation under random arrival with known and even $n$.
	Then $n\ge 2$ and, therefore, $\alg$ completes at least one phase.
	Let $i^*$ be the largest index such that $\ialg$ completes Phase~$i^*$.
	We claim that $2^{i^*+1}\ge \frac n4$.
	Assume for contradiction that $2^{i^*+1} < \frac n4$.
	Then, the number of agents that have arrived until the completion of iteration $i^*$ is $\sum_{i = 0}^{i^*}2^{i+1} \le 2 \cdot 2^{i^*+1} < \frac n2$.
	Hence, there are still $2^{i^*+2}$ agents left to complete another iteration, a contradiction.
	
	Let $J\subseteq N$ be the random subset of agents in the last completed iteration.
	We have just shown that
	\begin{equation}\label{eq:Jsize}
		|J| \ge \frac n4\text.
	\end{equation}
	
	Moreover, it holds that 
	\begin{align*}% 
		&\EV_{\sigma\sim\Sigma(N)}\left[\SW[\ialg(G,\sigma)]\right]\\
		&\ge \EV_J\left[\EV_{\sigma\sim\Sigma(J)}\left[\SW[\alg(G{[J]},\sigma)]\right]\right]\\
		&\ge \EV_J[ c\ \SW[\pi^*({G[J]})]] = c\ \EV_J[\SW[\pi^*({G[J]})]]\text.
	\end{align*}
	
	There, $J$ is sampled uniformly at random from the set of agents.
	Our key insight is that the maximum objective value achieved by a random subset of agents is good compared to the maximum objective value for the whole game.
	Therefore, let $\pi^* := \pi^*(G)$ be a partition for $G$ achieving maximum welfare.
	Recall that, for a given subset $J\subseteq N$, $\pi^*[J]$ denotes the partition $\pi^*$ restricted to $J$.
	Moreover, define $E^* := \left\{\{u,v\}\in {N\choose 2}\colon u\in \pi^*(v)\right\}$, i.e., the pairs of agents that are in a joint coalition in $\pi^*$.
	Note that for every set $\{u,v\}\in E^*$, it holds that 
	
	$$\pr\left(\{v,w\}\subseteq J\right) =\frac{{{|J|}\choose 2}}{{n\choose 2}} \ge \frac{{{\frac n4}\choose 2}}{{n\choose 2}} = \frac 1{16} \frac{n-4}{n-1}\text.$$
	
	For the first equality, we use that every pair of agents occurs in $J$ with equal probability.
	In the inequality, we applied \Cref{eq:Jsize}.
		
	Assume now first that $n\ge 7$.
	Then, $ \frac 1{16} \frac{n-4}{n-1}\ge \frac 1{32}$, and it follows that
	\begin{align*}
		\EV_J[\SW[\pi^*({G[J]})]]&\ge \EV_J[\SW[{\pi^*[J]}]] = \sum_{\{u,v\}\in E^*}\pr\left(\{v,w\}\subseteq J\right)2 w(u,v)\\
		& \ge \sum_{\{u,v\}\in E^*} \frac 1{32}2 w(u,v) = \frac 1{32}\sum_{\{u,v\}\in E^*} 2 w(u,v) = \frac 1{32}\SW[\pi^*]\text.
	\end{align*}
	
	Together with our previous computations, this yields
	$\EV_{\sigma\sim\Sigma(N)}\left[\SW[\ialg(G,\sigma)]\right] \ge \frac c{32}\SW[\pi^*]$
	
	Now, consider the case where $n\le 6$.
	Note that $\alg$, when executed for instances where the information is given that $n = 2$, has to form a coalition of size~$2$ if the two arriving agents have a positive utility.
	Hence, in this case, $\EV_{\sigma\sim\Sigma(N)}\left[\SW[\ialg(G,\sigma)]\right]$ is at least twice the utility of a positive edge selected uniformly at random.
	We conclude that then 
	$$\EV_{\sigma\sim\Sigma(N)}\left[\SW[\ialg(G,\sigma)]\right]\ge \frac 1{15} \SW[\pi^*] \ge \frac c{15} \SW[\pi^*]
	\ge \frac c{32} \SW[\pi^*]\text.$$
	
	There, we have used that $c\le 1$.
	Hence, combining both cases, we conclude that 
	$$c_{\ialg} \ge \frac c{32}\text.\eqno\qedhere$$
\end{proof}

The lemma implies that the competitive ratio of $\iew$ satisifies $\Omega\left(\frac{1}{n}\right)$.
We also obtain a matching upper bound by once again using the game in \Cref{fig:hard_instance}.
The analysis is even more involved as in \Cref{thm:wgdy} because we need to consider dependencies of the distributions of the agents between the phases of $\iew$.
We defer the analysis to \Cref{app:UppBounds}.

\begin{restatable}{theorem}{thmIEW}\label{thm:IEW}
	$\iew$ has a competitive ratio of $\Theta\left(\frac{1}{n}\right)$.
\end{restatable}

In conclusion, we have shown that we can achieve a competitive ratio of $\Theta\left(\frac{1}{n}\right)$ for coalition formation under random arrival without further restrictions like using randomization or knowing the number of agents.

\subsection{Obstacles for Optimal Algorithms}
In the second part of this section, we will discuss obstacles for algorithms with an optimal competitive ratio.
First, we show that no randomized algorithm beats a competitive ratio of $\frac 13$.
It is an interesting open question whether $\iew$ achieves an asymptotically optimal competitive ratio.
We conjecture that it is at least true that no algorithm achieves a constant competitive ratio.

Our proof relies on establishing a connection to the matching domain.
In this domain, there exists no $ c$-competitive algorithm with $c > \frac 13$ if the number of agents is unknown to the algorithm \citep{BRS25a}.
Interestingly, this result holds on a restricted domain of games, where positive weights are sparse:
A symmetric ASHG is said to belong to the \emph{tree domain} if every connected component of the edges with positive weight in the
associated undirected graph forms a tree,
and every edge with negative weight is larger in absolute value than
the sum of all positive edge weights.

\begin{theorem}[\citep{BRS25a}]\label{thm:negativeresult_MWMrand}
	In the random arrival model,  no online matching algorithm\footnote{The theorem by \citet{BRS25a} even holds for randomized algorithms, which would entail that we cannot even beat a competitive ratio of $\frac 13$ for online ASHGs when using randomization.
	For consistency with our remaining paper, we only state the result for deterministic algorithms.} has a competitive ratio of more than
	$\frac{1}{3}$ on the tree domain.
\end{theorem}

However, we can show that this immediately entails an impossibility result for online coalition formation algorithms.
We remark that \citet{BRS25a} prove an analogous result for online fractional hedonic games.

\begin{proposition}\label{prop:treeinherit}
	Consider the random arrival model and let $c\le 1$. 
	Assume that no $c$-competitive online matching algorithm exists for the tree domain.
	Then there exists no $c$-competitive online coalition formation algorithm.
\end{proposition}

\begin{proof}
	Let $c\le 1$. 
	Assume that we are given a $c$-competitive online coalition formation algorithm $\alg$.
	We construct a $c$-competitive online matching algorithm $\alg'$ on the tree domain that never forms a coalition of size~$3$ or more.
	To this end, let $\alg'$ simulate $\alg$, i.e., whenever a new agent and her valuations are revealed to $\alg'$, it feeds the same input to $\alg$.
	Then, $\alg'$ observes the output of $\alg$.
	If the output is a matching, $\alg'$ proceeds with this partition.
	However, once a coalition of size at least~$3$ is formed, $\alg'$ forms a singleton coalition instead and puts all remaining agents into singleton coalitions. 
	Clearly, $\alg'$ returns matchings and, therefore, is a matching algorithm.
	
	Moreover, on the tree domain, $\alg'$ achieves at least as high expected welfare as $\alg$.
	Indeed, on this domain, every coalition of size at least~$3$ contains an agent pairs such that their negative weight is less than all positive weights of the whole partition.
	Hence, by avoiding partitions with such coalitions, the expected welfare increases.
	Note that $\alg'$ operates like $\alg$ while not forming coalitions of size at least~$3$, and, therefore, achieves the same expected welfare for such arrival orders.
	Thus, $\alg'$ is $c$-competitive on the tree domain against all possible partitions and, therefore, in particular, against all matchings.
\end{proof}

Note that \Cref{prop:treeinherit} also holds for other models of online coalition formation, such as the standard model with adversarial arrival or the free dissolution model. 
We omit a deeper discussion because we only apply it under random arrival.
Combining \Cref{thm:negativeresult_MWMrand} and \Cref{prop:treeinherit}, we obtain the following result.

\begin{proposition}
	In the random arrival model, no online coalition formation algorithm has a competitive ratio of more than~$\frac 13$.
\end{proposition}

Finally, we will discuss general obstacles to finding classes of instances on which all algorithms perform poorly.
Interestingly, the performance of $\gdy$, $\wgdy$, and $\iew$ is bounded by the same class of instances, namely the instances illustrated in \Cref{fig:hard_instance}.
This raises the question whether the competitive ratio of any algorithm is $\Theta\left(\frac 1n\right)$ on these instances.
Our following result shows that this is not the case.
Indeed, for small enough $\epsilon$, all instances in this family have the property that there exists a unique highly valuable edge, whose weight achieves a constant fraction of the maximum social welfare. 
However, on instances of this type, we can use an optimal stopping algorithm to achieve a good competitive ratio.
To this end, we show how to apply the odds algorithm \citep{Brus00a}.

We only sketch the most important insights here and defer the details to \Cref{ap:odds}.
The goal is to design an algorithm that matches the maximum weight edge with high probability.
The idea of stopping algorithms is to only perform a decision, i.e., in our case to form a non-singleton coalition, once a certain stopping criterion is reached.
Usually, this requires knowledge of the number of agents~$n$. 

However, we can apply \Cref{lem:doubling} to obtain an algorithm for unknown~$n$ whose competitive ratio is only worse by a constant factor.
This has the effect that we are running a stopping algorithm in phases of an increasing number of agents.
Note that this does not have any implications on optimal stopping algorithms for unknown~$n$, which do not exist without further assumptions \citep[Chapter~4]{Brus00a}.
By constrast, in our iterated doubling variant, we may stop and match a good edge several times, which is not allowed in a stopping algorithm.

Let us briefly discuss the stopping conditions of our algorithm.
A stopping algorithm aims at achieving a success event among independent events.
As the $k$th event, we consider the event that the maximum weight edge connected to the $k$th agent is strictly larger than all edges among the first $k-1$ agents.
By the odds theorem  \cite[Theorem~1]{Brus00a}, we find a condition to stop at the last successful event, which in our case would reveal the maximum weight edge.
Since this condition is achieved with constant probability, this allows us to match the maximum weight edge with constant probability.

\begin{restatable}{theorem}{PropMaxe}\label{prop:maxe}
	Let $I$ be a set of ASHGs and $\lambda \in (0, 1]$ a constant such that for each ASHG in $I$ given by $G = (N,w)$, there is a unique maximum weight edge $e_{\max}$ with weight $w(e_{\max}) \ge \lambda \cdot \SW[\pi^*(G)]$, where $\pi^*(G)$ maximizes social welfare. 
	Then there exists an online coalition formation algorithm $\alg$ with $c_\alg \in \Theta\left(1\right)$ on $I$ in the random arrival model.
\end{restatable}

Notably, the obtained competitive ratio is a rather small constant, parameterized by $\lambda$.
For the case that the agent number~$n$ is known, the competitive ratio of the algorithm designed in the proof of \Cref{prop:maxe} is at least $\frac {\lambda}e$, where $e$ is Euler's number.
Hence, after the application of \Cref{lem:doubling}, we obtain a competitive ratio of at least $\frac {\lambda}{32e}$ for unknown~$n$.

\section{Free Dissolution Model}

This section considers online algorithms for coalition formation under free dissolution.
In this model, the agents arrive in a deterministic order, but the algorithm has the option to dissolve formed coalitions into singletons to rectify past decisions.
First, we consider optimal matching algorithms.
Recall that these are special coalition formation algorithms, where all coalition sizes are bounded by~$2$ throughout.

\subsection{General Online Matching under Free Dissolution}

We start by analyzing the capabilities of online matching algorithms in our model of weighted and non-bipartite instances.
Again, we consider the greedy algorithm a benchmark for more sophisticated algorithms.
For bipartite online matching instances, where one side of the agents is present offline, \citet{FKM+09a} show that $\gdy$, i.e., transitioning from $\mu$ to best \emph{matchings} within $\mathcal A^D(\genmat,i)$, achieves a competitive ratio of $\frac 12$.
However, $\gdy$ does not achieve a constant competitive ratio if all vertices arrive online.

\begin{restatable}{proposition}{gdymatching}\label{thm:gdy-disposal}
	When allowing free dissolution, $\gdy$ has a competitive ratio of $\Theta\left(\frac 1n\right)$ in the matching domain.
\end{restatable}

\begin{proof}
	First, note that $\gdy$ will match a pair of agents with the highest possible weight $W_{\max}$.
	Moreover, every matching has weight at most $\frac n2W_{\max}$, and, therefore, $c_\gdy \in \Omega\left(\frac 1n\right)$.
	
	\begin{figure}[tbp]
		\centering
		\resizebox{.9\textwidth}{!}{
			\begin{tikzpicture}
				\foreach \i/\k in {0/0,1/1,2/2,3/3,4/{k},5/{k+1}}
				{
					\node[draw, circle](a\i) at (3*\i,0){\color{white}{$a33$}};
					\node at (3*\i,0){$a_{\k}$};
				}
				\draw (a0) edge node[pos = 0.5, fill = white]{$1$} (a1);
				\draw (a1) edge node[pos = 0.5, fill = white]{$1+\epsilon$} (a2);
				\draw (a2) edge node[pos = 0.5, fill = white]{$1+2\epsilon$} (a3);
				\draw (a4) edge node[pos = 0.5, fill = white]{$1+k\epsilon$} (a5);
				\node at (barycentric cs:a3=1,a4=1) {\dots};
			\end{tikzpicture}
		} 
		\caption{Family of instances for the upper bound of the competitive ratio of $\gdy$ for matching instances under free dissolution.\label{fig:gdy-disposal}}
	\end{figure}
	
	For the upper bound, consider the family of instances depicted in \Cref{fig:gdy-disposal}.
	Let $\epsilon > 0$ and $k\in \mathbb N$ even.
	Consider $G^{k,\epsilon} = (N^{k,\epsilon}, w^{k,\epsilon})$ where $N^{k,\epsilon} = \{a_i\colon 0 \le i \le k+1\}$ and, for $i\in [k+1]$, we set $w^{k,\epsilon}(a_{i-1},a_i) = 1 + (i-1)\epsilon$.
	All other weights are set to $0$.
	Note that $\genmat^*(G^{k,\epsilon}) := \left\{\{a_{2i},a_{2i+1}\}\colon 0\le i \le \frac{k}{2} \right\}$ is the maximum weight matching for $G^{k,\epsilon}$.
	
	Now, let the arrival order be $(a_0,a_1,\dots, a_k,a_{k+1})$.
	Then, $\gdy$ outputs the matching $\{\{a_k,a_{k+1}\}\}$.
	Hence,
	$$
		c_{\gdy}\le \inf_{\epsilon > 0}\frac{\SW[\gdy(G^{k,\epsilon},\sigma^k)]}{\SW[\genmat^*(G^{k,\epsilon},\sigma^k)]}
		= \inf_{\epsilon > 0}\frac{1 + k\epsilon}{\frac k2 +\frac 14k(k+2)\epsilon} = \frac 2k \in \Theta\left(\frac 1n\right)\text.\eqno\qedhere
	$$
\end{proof}

The non-zero edges in the construction of the previous proposition even induce \emph{bipartite} graphs.
Hence, the competitive ratio of $\frac 12$ of the greedy algorithm on bipartite instances if one side of the agents is present offline \citep{FKM+09a} relies on the existence of offline agents and not on bipartiteness.
By constrast, \Cref{thm:gdy-disposal} indicates that even achieving a constant competitive ratio is a non-trivial task for our setting of online matching.

Our next goal is to show how to achieve a constant competitive ratio. 
Even more, we will show how to obtain the best possible competitive ratio.
To this end, we slightly modify the algorithm by \citet{McGr05a} for the edge arrival model.
The key idea is to adjust the greedy algorithm by only allowing edge dissolution if this achieves sufficient gain. 

\begin{definition}[Dissolution threshold algorithm with parameter $t$]
	Let $t\ge 1$ be a parameter. Given a matching $\genmat$ and a newly arriving vertex $i$, the \emph{dissolution threshold algorithm with parameter $t$} ($t$-$\ggdy$) is the greedy algorithm for the available matchings 
	\begin{equation*}
		\avai(\genmat,i) = 
		\big\{(\genmat - j)\cup \{\{i,j\}\} \colon \{j\}\in \genmat\big\} 
		\cup 
		\big\{(\genmat \setminus \{C\})\cup \{\{i,j\},\{\ell\}\} \colon C = \{j,\ell\}\in \genmat, w(i,j)\ge t\, w(j,\ell)\big\}\text.
	\end{equation*}
\end{definition}

Hence, $t$-$\ggdy$ has two possibilities.
First, it can form a new edge with a singleton vertex.
This is captured in the first set and corresponds to chosing a matching in $\avai^S(\genmat,i)$ where $i$ is matched.
Second, it can dissolve a matched pair if the weight of the new edge compared to the weight of the dissolved edge is larger by a factor of at least~$t$.
Under these possibilities, it makes a greedy decision, i.e., it chooses the best option if the weight of the matching can be increased, and leaves the vertex unmatched otherwise.

Specifically, we are interested in $t$-$\ggdy$ for $t = 1 + \frac {\sqrt{2}}2$.
As we will soon see, this achieves an optimal competitive ratio, not only among $t$-$\ggdy$ algorithms but even among all possible deterministic algorithms.
The proof of the following theorem only requires straightforward adaptations of the techniques by \citet{FKM+05a} and \citet{McGr05a} who obtain the analogous result for the edge arrival model.\footnote{In the vertex arrival model, whenever a vertex arrives, a whole set of edges, i.e., the edges incident to this vertex arrive.
	Hence, we essentially have to prove that this type of simultaneous arrival of edges does not allow for better algorithms.
In the conference version of our paper, we were not yet aware of the close connection of this work to online matchings and had independently analyzed $2$-$\ggdy$ with the same approach as \citet{FKM+05a}.}
For the sake of completeness, we include the proof in the appendix.

\begin{restatable}{proposition}{THMggdy}\label{thm:ggdy}
	When allowing for free dissolution, $\left(1 + \frac {\sqrt{2}}2\right)$-$\ggdy$ has a 
	competitive ratio of $\frac 1{3 + 2\sqrt{2}} \approx 0.172$ in the matching domain.
\end{restatable}

Surprisingly, $\left(1 + \frac {\sqrt{2}}2\right)$-$\ggdy$ achieves an optimal competitive ratio. 
Notably, it is less than $\frac{1}{2}$, which is a typical optimal competitive ratio in more restricted online matching instances \citep{KVV90a,FKM+09a,WaWo15a}.
This is a direct consequence of a more general result by \citet{Vara11a}, who studied the online buyback problem under $k$ matroid constraints.
In this setting, an online algorithm receives a series of bids for items it can sell, subject to $k$ matroid constraints.
The decision not to sell an item is irrevocable.
However, sold items can be bought back for a fraction $f$ of the bid price.
It is possible to reduce our setting to the online buyback problem in worst-case instances.

First, we have the case $f = 0$ because we can dissolve without incurring any cost.
Next, consider tree graphs and note that every tree graph is bipartite.
It is well known that matchings in bipartite graphs can be described with $2$ matroid constraints \citep[see, e.g.,][]{KoVy06a}.
Therefore, we can sell edges in a tree graph with $2$ matroid constraints such that the sold edges form a matching.
We are now already very close to our setting, except that edges, and not vertices arrive over time.
But, if we additionally require that the agents arrive in an order such that at each step, the current graph is a tree, then edge and vertex arrival are equivalent: except for the first edge, the arrival of an edge coincides with the arrival of exactly one new vertex; otherwise, the new edge would be disconnected or form a cycle.
Now, all instances in his proof concerning the worst-case behavior in the setting for $k = 2$ by \citet{Vara11a} are matching instances on tree graphs.
Hence, Theorem 1 by \citet{Vara11a} for the case $k = 2$ and $f = 0$ implies that there is no $c$-competitive algorithm, with $c > \frac{1}{3 + 2\sqrt{2}}$.

\begin{proposition}
	In the deterministic arrival model, no online matching algorithm achieves a competitive ratio of more than $\frac{1}{3 + 2\sqrt{2}}$ under free dissolution.
	In particular, $\left(1 + \frac {\sqrt{2}}2\right)$-$\ggdy$ achieves an optimal competitive ratio.
\end{proposition}

\subsection{Online Coalition Formation}

Departing from online matching, we now analyze the capabilities and boundaries of online coalition formation algorithms under free dissolution.
As a benchmark, we again start with an analysis of the greedy algorithm.
In \Cref{thm:gdy-disposal}, we have seen that $\gdy$ has a competitive ratio of $\Theta\left(\frac 1n\right)$ in the matching domain.
In the coalition formation domain, $\gdy$ performs even worse.

\begin{restatable}{proposition}{gdydisscofo}\label{prop:gdydisscofo}
	When allowing for free dissolution, $\gdy$ has a competitive ratio of $\Theta\left(\frac 1{n^2}\right)$ in the coalition formation domain.
\end{restatable}

\begin{proof}
	First, we prove the lower bound for the competitive ratio of $\gdy$. 
	Consider an arbitrary ASHG given by $G = (N,w)$ together with an arrival order $\sigma$.
	Define $E^+(G) := \left\{e\in {N \choose 2}\colon w(e) > 0\right\}$, which is the set of agent pairs with positive weights.
	Then, the optimal partition $\pi^*(G)$ satisfies $\SW[\pi^*(G)] \le 2w(E^+(G))$. 
	
	In addition, because of the free dissolution, $\gdy$ maintains a partition whose social welfare is at least as high as the social welfare of the coalition containing just two agents connected by a maximum weight edge.
	Indeed, when the second agent of such a pair arrives, $\gdy$ has the opportunity to achieve at least this welfare by dissolving the coalition of the first agent of the pair.
	Hence, in this step, a partition is reached that achieves at least this welfare.
	Consequently, 
	\begin{equation*}
		\SW[\gdy(G,\sigma)] \ge \max_{e\in E^+(G)}2 w(e) \ge \frac 2{|E^+(G)|}\sum_{e\in E^+(G)}w(e)\ge \frac 2{n^2}w(E^+(G))\ge \frac 1{n^2}\SW[\pi^*(G)]\text.
	\end{equation*}
	Thus, the competitive ratio of $\gdy$ satisfies $\Omega\left(\frac{1}{n^2}\right)$. 
	
	To show that $\gdy$ performs no better than that, we now construct a family of instances on which it performs poorly.
	Let $k\in \mathbb N$ with $k\ge 2$.
	We construct a game $G^k = (N^k,w^k)$ and illustrate it in \Cref{fig:GDYdiss}.
	
	\begin{figure}
		\centering
		\begin{tikzpicture}
			\pgfmathsetmacro\figbreadth{2.5}
			\foreach[count = \k] \i in {4.5,3,1.5}
			{
			\foreach \j/\a in {0/a,\figbreadth/b}
			{
			\node[draw, circle] (\a\k) at (\j,\i) {$\a_{\k}$};
			}
			}

			\node[draw, circle] (a4) at (0,-1.5) {$a_k$};
			\node[draw, circle] (b4) at (\figbreadth,-1.5) {$b_k$};

			\node[rotate = 90] at (0,0) {$\dots$};
			\node[rotate = 90] at (\figbreadth,0) {$\dots$};
			
			\draw (a1) edge node[fill = white, pos = .1] (u4) {\color{white}$n$} (b4);
			\draw (a1) edge node[fill = white, pos = .15] (u3) {\color{white}$n$} (b3);
			\draw (a1) edge node[fill = white, pos = .2] (u2) {\color{white}$n$} (b1);
			\draw (a1) edge node[fill = white, pos = .2] (u1) {\color{white}$n$} (b2);
			
			\foreach \i in {2,3,4}
			{
			\foreach \j in {1,2,3,4}
			{\draw (a\i) -- (b\j);
			\node at (u\j) {\tiny$-k$};
			}
			}

			\draw (a1) -- (a2);
			\draw[bend right = 40] (a1) edge node[fill = white, pos = .5] {\footnotesize$\epsilon$} (a3);
			\draw[bend right = 40] (a1) edge node[fill = white, pos = .5] {\footnotesize$\epsilon$} (a4);
		\end{tikzpicture}
		\caption{Worst-case performance of $\gdy$ with free dissolution in the coalition formation domain. Missing edges represent a utility of~$0$. Unlabelled edges represent a utility of~$1$.}\label{fig:GDYdiss}
	\end{figure}

	Define $N^k = \{a_i,b_i\colon 1\le i \le k\}$, i.e., $G^k$ has $2k$ agents.
	We set $\epsilon = \frac 1k$.
	Define weights by
	\begin{itemize}
		\item $w(a_1,a_2) = 1$,
		\item $w(a_1,a_i) = \epsilon$ for $3\le i\le k$,
		\item $w(a_1,b_j) = -k$ for $1\le j \le k$,
		\item $w(a_i,b_j) = 1$ for $2\le i\le k$ and $1\le j\le k$, and
		\item all other weights are set to $0$.
	\end{itemize}
	
	Assume that agents arrive in the order $\sigma^k = (a_1, \dots, a_k, b_1, \dots, b_k)$.
	Then, $\gdy$ forms a big coalition containing $C := \{a_i\colon 1\le i \le k\}$ during the arrival of the first $k$ agents.
	Then, for every $1\le j\le k$, at her arrival, $b_j$ cannot join $C$ because this would yield a negative improvement in social welfare.
	Also, $\gdy$ does not dissolve $C$ and form a coalition of an agent in $C$ with $b_j$ because this only yields a social welfare of $2$, while $C$ already achieved a social welfare of $2 + 2(k-2)\epsilon$.
	Hence, $b_j$ forms a new singleton coalition.
	It follows that $\SW[\gdy(G^k,\sigma^k)] = 2 + 2(k-2)\epsilon < 4$.
	
	Consider, however, the partition maximizing social welfare is $\pi^* = \{\{a_2,\dots, a_k,b_1,\dots, b_k\},\{a_1\}\}$.
	Note that $\SW[\pi^*] = 2(k-1)k = \Theta(n^2)$.
	It follows that $c_\gdy \in \orderof{\frac 1{n^2}}$.
\end{proof}

Once again, a competitive ratio of $\Theta\left(\frac 1{n^2}\right)$ is bad in the coalition formation domain: 
essentially, it means that, in the worst case, the greedy algorithm can only achieve the welfare accumulated by one good pair of agents. 
This again raises the question of whether we can find algorithms with a better competitive ratio.
A simple solution is to apply
$\left(1 + \frac {\sqrt{2}}2\right)$-$\ggdy$ as a coalition formation algorithm.
By \Cref{lem:mat2cofo}, we obtain an $\Omega\left(\frac 1n\right)$-competitive algorithm under free dissolution.

\begin{corollary}
	When allowing for free dissolution, $\left(1 + \frac {\sqrt{2}}2\right)$-$\ggdy$ is $\frac{1}{(3+2\sqrt{2})n}$-competitive for coalition formation .
\end{corollary}

Note that every decision regarding how to update the matching in $\left(1 + \frac {\sqrt{2}}2\right)$-$\ggdy$ can be made in polynomial time.
However, because of the inapproximability result by \citet{FKV22a},
we do not expect \emph{efficient} online algorithms to achieve an asymptotically better competitive ratio.
In fact, as our following theorem shows, online coalition formation algorithms cannot achieve more even with unlimited computational power.

The proof idea is to iteratively construct an 
adversarial instance whose edge weights are defined based on previous decisions of the algorithm.
The magnitude of the weights increases by a factor of $n$ for every new agent.
This forces the algorithm to form coalitions with new agents.
However, by assigning positive and negative weights to existing coalitions, the algorithm is essentially forced to maintain coalitions of size at most~$2$ (or it will miss a significant gain of social welfare).
Consequently, the algorithm can only achieve a social welfare that is at most four times the maximum weight.
By contrast, the optimal partition can achieve a welfare of linear order with respect to the maximum weight.

\begin{restatable}{theorem}{DissolutionStrongUB}\label{thm:DissolutionStrongUB}
	In the deterministic arrival model, the competitive ratio of any online coalition formation algorithm under free dissolution is at most~$\frac {12}n$.
\end{restatable}

\begin{proof}
	In this proof, we repeatedly bound the partial sum of the geometric sequence as 
	\begin{equation}\label{eq:geometric}
		\sum_{k=0}^{n}x^k = \frac{x^{n + 1} - 1}{x - 1} \le x^{n + 1} \text{ for } x \ge 2\text.
	\end{equation}
	
	Let $\alg$ be any online coalition formation algorithm under free dissolution.
	Let $n\in \mathbb N$ with $n\ge 12$.
	We will iteratively construct an adversarial instance $(G^n,\sigma^n)$ where $G^n = (N^n, w^n)$ is a graph with $n$ agents, for which it holds that 
	$$\SW[\alg(G^n,\sigma^n)] \le \frac {12}n \SW[\pi^*(G^n)]\text.$$
	
	For the remaining proof, we will omit the constructed instances' superscript for a more straightforward exposition.
	Let $N = \{a_i\colon i\in [n]\}$ be the agent set, and assume that agents arrive in the order $\sigma = (a_1,\dots, a_n)$.
	For $i\in [n]$, let $\pi_i$ denote the partition created by $\alg$ on input $(G,\sigma)$ after the arrival of $a_i$.
	Clearly, $\pi_1 = \{\{a_1\}\}$.
	
	Let $i\ge 2$.
	For our adversarial instance, we distinguish two cases to define the weights of $a_i$ based on previously formed partitions.
	If there exists $j\in [i-1]$ with $\max\{|C|\colon C\in \pi_j\} \ge 3$, i.e., there exists a coalition of size at least~$3$ in a previous partition, then $w(a_i,a_j) = 0$ for all $j\in [i-1]$.
	Otherwise, it holds that $|C| \le 2$ for all $C\in \pi_{i-1}$.
	Consider any set $S_i \subseteq N$ with $|S_i\cap C| = 1$ for all $C\in \pi_{i-1}$ with $|C| = 2$, and $S_i\cap C = \emptyset$, otherwise.
	In other words, $S_i$ precisely contains one agent from every coalition of size~$2$. 
	We define 
	$$w(a_i,a_j) = \begin{cases}
		- n^i & \text{if}\ j \in S_i\text,\\
		n^i & \text{otherwise.}
	\end{cases}$$
	
		Thus, as long as the algorithm maintains partitions with coalitions of size at most~$2$, agents' utilities are increasing by a factor of $n$ for each agent. 
		Moreover, the $i$th arriving agent has a utility of $-n^i$ for one agent from each coalition of size~$2$ existing at her arrival and a utility of $n^i$ for all other agents. 
		In \Cref{fig:worst-case-diss}, we display a possible graph evolving from our construction.
		For a clearer exposition, we display the hypothetical situation where $n = 5$, even though our proof assumes $n\ge 12$.
		The displayed graph would be constructed if $S_2 = \emptyset$, $S_3 = \{2\}$, $S_4 = \{1\}$, and $S_5 = \{1,2\}$.
		This might be the adversarial instance if $\alg$ forms $\pi_2 = \{\{a_1,a_2\}\}$, then $\pi_3 = \{\{a_1,a_3\}, \{a_2\}\}$ after dissolving $\{a_1,a_2\}$, and afterwards $\pi_4 = \{\{a_1,a_3\},\{a_2,a_4\}\}$.
	
	\begin{figure}
		\centering
		\begin{tikzpicture}
			\node[draw, circle] (a3) at (0,0) {$a_3$};
			\node[draw, circle] (a1) at (150:2.1) {$a_1$};
			\node[draw, circle] (a2) at (30:2.1) {$a_2$};
			\node[draw, circle] (a4) at (0,-2.1) {$a_4$};
			\node[draw, circle] (a5) at (0,-4.2) {$a_5$};
		
			\draw (a1) edge node[fill = white, pos = .5] {\footnotesize $5^2$} (a2);
		
			\draw (a1) edge node[fill = white, pos = .5] {\footnotesize $5^3$} (a3);
			\draw (a2) edge node[fill = white, pos = .5] {\footnotesize $-5^3$} (a3);

			\draw (a1) edge node[fill = white, pos = .5] {\footnotesize $-5^4$} (a4);
			\draw (a2) edge node[fill = white, pos = .5] {\footnotesize $5^4$} (a4);
			\draw (a3) edge node[fill = white, pos = .5] {\footnotesize $5^4$} (a4);

			\draw[bend right] (a1) edge node[fill = white, pos = .5] {\footnotesize $-5^5$} (a5);
			\draw[bend left] (a2) edge node[fill = white, pos = .5] {\footnotesize $-5^5$} (a5);
			\draw[bend right= 35] (a3) edge node[fill = white, pos = .5] {\footnotesize $5^5$} (a5);
			\draw (a4) edge node[fill = white, pos = .5] {\footnotesize $5^5$} (a5);
		\end{tikzpicture}
		\caption{Illustration of the adversarial instances in the proof of \Cref{thm:DissolutionStrongUB}.
		\label{fig:worst-case-diss}}
	\end{figure}
	
	We compute a bound for the performance of $\alg$ for both cases in the definition of the adversarial instance separately.
	We start with the case where $\alg$ forms a coalition of size at least $3$.
	
	\begin{claim}\label{cl:largecoals}
		Assume that there exists $j\in [n]$ with $\max\{|C|\colon C\in \pi_j\} \ge 3$.
		Then, $$\SW[\alg(G,\sigma)] \le \frac 2n \SW[\pi^*(G)]\text.$$
	\end{claim}
	
	\begin{claimproof}
		Assume that there exists $j\in [n]$ with $\max\{|C|\colon C\in \pi_j\} \ge 3$ and let $j^*$ be the minimum integer for which this is the case.
		
		Then, for $\ell > j^*$, the utilities between $a_\ell$ and any other agent is~$0$.
		Additionally, $a_{j^*}$ joins a coalition of size~$2$, which does not increase the social welfare by construction.
		Therefore, $\SW[\alg(G,\sigma)]$ is bounded by the largest social welfare that can be achieved by $\pi_{j^*-1}$.
		
		Moreover, by assumption on $j^*$, partition $\pi_{j^*-1}$ consists of coalitions of size at most~$2$.
		There, the utility of the agents in each coalition is bounded by the utility of the higher indexed agent, where each of the utilities $n^\ell$ for $2\le \ell \le j^*-1$ is assumed by at most one pair.
		Consequently, 
		
		$$\SW[\alg(G,\sigma)] \le \sum_{\ell = 2}^{j^*-1} 2n^\ell \le 2n^{j^* - 1} + 2\sum_{\ell = 0}^{j^* - 2}n^\ell \overset{\text{Eq.~(\ref{eq:geometric})}}{\le} 4n^{j^*-1}\text.$$
		
		However, we can construct a partition with comparatively high welfare, which gives a lower bound on the value of the maximum welfare partition.
		Let $i^* \in [j^* - 1]$ with $w(a_{j^*},a_{i^*}) = n^{j^*}$.
		Such an agent exists as $a_{j^*}$ has a negative utility for at most half of the agents present at her arrival.
		
		Consider the partition $\pi' = \{\{a_{j^*},a_{i^*}\}\} \cup \{\{a_k\}\colon k\in [n]\setminus \{j^*,i^*\}\}$.
		It holds that $\SW[\pi'] = 2n^{j^*}$.
		We obtain
		$$\SW[\alg(G,\sigma)] \le 4 n^{j^*-1} = \frac 2n 2n^{j^*} = \frac 2n \SW[\pi'] \le \frac 2n \SW[\pi^*(G)]\text.\eqno\qedhere$$
	\end{claimproof}
	
	It remains to consider the case where the algorithm maintains coalitions of size at most~$2$.

	\begin{claim}\label{cl:smallcoals}
		Assume that, for all $j\in [n]$, it holds that $\max\{|C|\colon C\in \pi_j\} \le 2$.
		Then, $$\SW[\alg(G,\sigma)] \le \frac {12}n \SW[\pi^*(G)]\text.$$
	\end{claim}
	
	\begin{claimproof}
		Assume that, for all $j\in [n]$, it holds that $\max\{|C|\colon C\in \pi_j\} \le 2$.
		Hence, in particular, $\pi_n$ consists of coalitions of size at most~$2$.
		By the design of the utilities, we know that
		
		\begin{equation*}
			\SW[\alg(G,\sigma)] \le \sum_{j =  \lceil \frac{n}{2} \rceil}^{n} 2n^j \le 2n^n + 2\sum_{j = 0}^{n - 1}n^j \overset{\text{Eq.~(\ref{eq:geometric})}}{\le} 4n^n\text.
		\end{equation*}
	
		Now, we once again construct a partition that performs significantly better in terms of welfare.
		Recall that $S_n$ is the set of agents for which $a_n$ has a negative utility.
		Consider the partition $\pi' = \{\{a_j: j \in [n] \setminus S_n\}\} \cup \{\{a_j\}: j \in S_n\}$, i.e., the only coalition of size larger than~$1$ is formed by $a_n$ and all agents for which $a_n$ has a positive utility.
		Let $T = N\setminus S_n$ and $t = |T| = |N\setminus S_n|$ be the set and number of agents that form a coalition together with $a_n$.
		Then, the negative utility accumulated by agents in $T$ is bounded by 
		\begin{equation*}
			\sum_{j = 1}^{t-1} 2(t-j) n^{n-j} \le 2(t - 1)\left(n^{n - 1} + \sum_{j=2}^{t-1}n^{n-j} \right) \le 2(t - 1)\left(n^{n - 1} + \sum_{j=0}^{n - 2}n^{j} \right) \overset{\text{Eq.~(\ref{eq:geometric})}}{\le} 4(t-1)n^{n-1}\text.
		\end{equation*}
	
		Hence, $\SW[\pi'] \ge 2(t-1) n^n - 4(t-1) n^{n-1}$. 
		Next, we combine these insights.
		In the following computations, the second inequality holds because of $t \le n - 1$, and the third holds whenever $\frac n3 \le t-1$. 
		The latter is true for $n\ge 3$ since $t\ge \frac{n-1}2$.
		Finally, the penultimate inequality is valid for $n\ge 12$.
		We obtain
		\begin{align*}
			& \SW[\alg(G,\sigma)] + 4(t-1)n^{n-1}
			 \le 4n^n + 4(t-1)n^{n-1} \le 8 n^n 
			= \frac {12}n 2 \frac {n}3 n^n \le \frac {12}n 2(t-1)n^n\\
			&\le \frac{12}n \SW[\pi'] + \frac {12}n4(t-1) n^{n-1}
			\le \frac{12}n \SW[\pi'] + 4(t-1) n^{n-1} 
			\le \frac{12}n \SW[\pi^*(G)] + 4(t-1) n^{n-1}\text.
		\end{align*}
		Subtracting $4(t-1) n^{n-1}$ from both sides of the chain of inequalities yields the desired bound.
	\end{claimproof}
	
	Combining Claims~\ref{cl:largecoals} and~\ref{cl:smallcoals}, we conclude that $\SW[\alg(G,\sigma)] \le \frac {12}n \SW[\pi^*(G)]$, as desired.
\end{proof}

\section{Conclusion}

\begin{table}
	\caption{Overview of online matching and coalition formation algorithms' competitive ratios.
	The rows display the performance of the greedy algorithm, the performance of the best-known algorithm, and the best-known upper bound for the competitive ratio of any algorithm.
	The general coalition formation model results are due to \citet{FMM+21a} and depend on the minimum and maximum utilities $\Umin$ and $\Umax$.
	In all our models, we bound the performance of the greedy algorithm independent of the utility range.
	In contrast to the general model, we can always beat the greedy algorithm with more sophisticated algorithms.
	\label{tab:results}}
			\begin{equation*}
			\renewcommand{\arraystretch}{1.7}
			\begin{array}{lcccc}
			& 
			& 
			&   \multicolumn{2}{c}{\textnormal{Free dissolution}}
			\\\cline{4-5}
			&\textnormal{Basic model} &\textnormal{Random arrival}& \multicolumn{1}{c}{\textnormal{Matching}} & \multicolumn{1}{c}{\textnormal{Coalition formation}}\\
			\midrule
			\textnormal{Greedy algorithm} &   \Theta\left(\frac 1n\frac{\Umin}{\Umax}\right) & \Theta\left(\frac 1{n^2}\right) &   \Theta\left(\frac 1n\right) & \Theta\left(\frac 1{n^2}\right)\\
			\textnormal{Best known algorithm} &   \Theta\left(\frac 1n\frac{\Umin}{\Umax}\right) & \Theta\left(\frac 1n\right) &   \frac 1{3 + 2\sqrt{2}} & \frac 1{(3 + 2\sqrt{2})n}\\% & & & \frac 16 \\
			\textnormal{General upper bound} &   \orderof{\frac 1n\frac{\Umin}{\Umax}} & \frac 13 & \frac 1{3 + 2\sqrt{2}} & \frac {12}n\\
			\bottomrule% 
			\end{array}
			\end{equation*}
\end{table}

We have considered two models of online coalition formation that both facilitate the existence of good algorithms compared to a deterministic arrival model.
In the first model, the adversary's power is reduced to merely selecting an instance, while the arrival order is random.
In the second model, algorithms must perform well against an adversary that can fix both the instance and the arrival order, but we have increased the algorithm's capabilities by allowing for free dissolution.
Our main results are summarized in \Cref{tab:results}.
Both models allow for algorithms that achieve a competitive ratio of $\Theta\left(\frac 1n\right)$.
Interestingly, this precisely eliminates weight dependencies of the best possible competitive ratio in the basic model studied by \citet{FMM+21a}.

Notably, a competitive ratio of $\Theta\left(\frac 1n\right)$ matches an inapproximability result by \citet{FKV22a}
for social welfare in the offline setting: unless $\P\,=\,\NP$, there exists no polynomial-time approximation algorithm for maximizing social welfare in the offline setting with an approximation ratio of $\frac 1 {n^{1-\epsilon}}$ where $\epsilon > 0$ is arbitrary.
Since every online algorithm can also be applied in the offline setting, this impossibility extends to online algorithms as well, i.e., there cannot be polynomial-time online algorithms with a competitive ratio of $\frac 1 {n^{1-\epsilon}}$.
Even though online algorithms are not necessarily bound to limited computational power, many online algorithms, and in particular, all algorithms considered in this paper, perform all decisions in polynomial time.
Hence, the existence of algorithms with a competitive ratio better than $\Theta \left( \frac 1n \right)$ seems unlikely.
In other words, we propose two online settings where our algorithmic capabilities are asymptotically as good as offline possibilities.
In the free dissolution model, we even prove a stronger statement: even with unlimited computational power, a competitive ratio of more than $\frac {12}n$ cannot be achieved.

In addition to improving the competitive ratio compared to the basic model, the nature of good algorithms also differs.
Recall that in the basic model, the greedy algorithm, which always assigns agents to coalitions leading to the highest immediate increase in social welfare, achieves an optimal competitive ratio.
By contrast, in both of our models, the competitive ratio of the greedy algorithm is worse than that of the best obtained algorithms by a linear factor. 

Matching plays a vital role in obtaining algorithms that improve upon the performance of the greedy algorithm.
In the random arrival model, we present an algorithm whose output dominates the weight of a simple matching algorithm.
Hence, matchings occur implicitly in the algorithm's analysis.
Under free dissolution, our coalition formation algorithm itself is a matching algorithm for a general online matching problem.
We use approaches from an online matching problem with edge instead of vertex arrival to analyze the matching domain.
The key algorithmic idea is to enhance the greedy algorithm by adding a threshold for the improvement in social welfare whenever dissolving a coalition (or edge).
In particular, this leads to an optimal competitive ratio for creating matchings.

Next, we want to discuss optimal algorithms.
In the random arrival model, deriving the optimal competitive ratio is left open
and our best bound is $\frac 13$.
It is an intriguing open problem to determine whether online algorithms with unlimited computational power can achieve an asymptotically better competitive ratio than $\Theta\left(\frac 1n\right)$, as obtained by our iterated waiting algorithm.
Interestingly, the worst-case instances for all of our considered algorithms under random arrival, i.e., the greedy algorithm, the waiting greedy algorithm for a known agent number $n$, as well as the iterated waiting algorithm for unknown $n$ are identical and have the property that instances contain a single important weight.
This raises the question of whether these instances lead to general algorithmic limitations.
We have answered this negatively by showing that an optimal stopping algorithm achieves a constant competitive ratio in a restricted domain containing these instances.

In the model allowing for free dissolution, we have shown that the competitive ratio we achieve is asymptotically optimal.
Moreover, we obtain interesting insights in the corresponding matching domain: there, the worst-case examples are trees, and, therefore, in particular bipartite.
Hence, it is already the combination of edge weights and having all vertices arriving online that causes algorithmic limitations. 

Notably, our paper was entirely focused on deterministic algorithms (which we see as a benefit of our algorithmic results).
Still, exploring whether randomization can lead to improved algorithms might be interesting. 
In particular, this is the case in the matching domain under free dissolution, where it is possible to beat the competitive ratio of $\frac 1{3 + 2\sqrt{2}}$ \citep{ELSW13a}.
We remark that the optimal competitive ratio achievable by randomized algorithms is still open in this setting.

Even now, the online model of coalition formation remains scarcely researched and seems to deserve further attention.
Aside from bridging the gap in the achievable competitive ratio under random arrival, it would be interesting to search for conditions that allow to beat a competitive ratio of $\Theta\left(\frac 1n\right)$, which still feels quite bad.
Because of the aforementioned inapproximability results, this is much related to gaining a better understanding of approximation guarantees in the offline model.
\citet{BCS25a} recently made progress in this question by showing that an approximation guarantee of $\orderof{\log n}$ can be obtained by a randomized polynomial-time algorithm for ASHGs in which the sum of all weights is non-negative.
Studying this and other weight restrictions could lead to new insights in the online model, as well.

Another avenue for future research would be to gain more knowledge about other classes of coalition formation games.
Some recent progress was presented by \citet{BRS25a}, who study random arrival and free dissolution in fractional hedonic games \citep{ABB+17a}.
However, hedonic games offer a wide range of further preference representations.
It would be interesting to study online algorithms in models where these are given by ordinal rankings rather than utilities, which can, for instance, be obtained by ranking coalitions according to their best or worst members \citep{CeRo01a}. 
An additional challenge in such models is to develop a reasonable quantifiable objective.

\section*{Acknowledgments}

This work was supported by the AI Programme of The Alan Turing Institute and the Deutsche Forschungsgemeinschaft under grants BR 2312/11-2 and BR 2312/12-1.
We thank Saar Cohen, Viktoriia Lapshyna, Alexander Schlenga, and Thorben Tr\"obst for valuable discussions.
We are grateful for the helpful comments by the anonymous reviewers of ESA 2023 and ACM TALG.

\publ{check whether page numbers are available for forthcoming articles. check if new journal versions have appeared}

\appendix
\section*{Appendix}

In the appendix, we provide more technical details of our results.

\section{Worst-Case Performance of Algorithms under Random Arrival}\label{app:UppBounds}

In this appendix, we provide the upper bounds for the performance of our algorithms under random arrival.
To focus on the crucial ideas, we defer the rather technical proofs of some binomial identities and inequalities to \Cref{app:binom}.

This appendix uses the following notation for a fixed arrival order $\sigma$.
Given an agent $x\in N$ and arrival time $t\in [n]$, we denote by $A_{\le t}^x$, $A_{=t}^x$, and $A_{\ge t}^x$ the events that $\sigma^{-1}(x)\le t$, $\sigma^{-1}(x) = t$ and  $\sigma^{-1}(x)\ge t$, respectively.

We now prove the upper bound for the performance of $\wgdy$.

\WGDY*

\begin{proof}[Proof of upper bound]
	We complete the proof by showing that $\wgdy$ has a competitive ratio of $\orderof{\frac{1}{n}}$. 
	To this end, we consider again the family of ASHGs given by $G^{k,\epsilon} = (N^{k,\epsilon},w^{k,\epsilon})$ as defined in the proof of \Cref{thm:compGDY} and illustrated in \Cref{fig:hard_instance}.
	Again, we analyze instances for fixed $k$ and omit $k$ as superscripts.
	
	Moreover, as in the proof of \Cref{thm:compGDY}, in the limit case for small $\epsilon$, the optimal partition has social welfare of $2$, and the social welfare of the partition computed by $\wgdy$ is equal to twice the probability of matching $a$ with $b$ in $G^{\epsilon^*}$ for $\epsilon^* = \frac 12$.

	We start by splitting the probability that $\{a, b\}$ forms a coalition by differentiating whether $a$ and $b$ arrive in the same phase of the algorithm and which one arrives first. 
	If they arrive in different phases, then the probability that the coalition $\{a,b\}$ forms is the same because of symmetry. 
	If both of them arrive in the first or second phase, then $\{a,b\}$ does not form, and thus, the probability is $0$. 
	Indeed, in the first phase, no coalitions are formed. 
	In the second phase, $a, b$ are directly matched to some agent $x \in X$ or $y \in Y$, respectively, because $k + 1$ agents arrive in the first phase and thus, because of the pigeonhole principle, at least one agent of each of $X$ and $Y$ arrives in the first phase.
	Thus, we can assume without loss of generality that $a$ arrives in the first phase and $b$ in the second. 
	The probability that $a$ arrives in the first phase is $\frac{1}{2}$ and then the probability that $b$ arrives in the second phase is $\frac{k + 1}{2k + 1}$.
	We obtain
	\begin{align*}
		c_\wgdy &\le \pr_\sigma\left(\{a,b\} \in \wgdy(G^{\epsilon^*}, \sigma)\right) \\
		&= \pr_\sigma\left(\{a,b\} \in \wgdy(G^{\epsilon^*}, \sigma) \ \Big|\  A^a_{\le k + 1}, A^b_{\ge k + 2}\right)\pr_\sigma\left(A^a_{\le k + 1}, A^b_{\ge k + 2}\right) \\
		&\phantom{=}\ + \pr_\sigma\left(\{a, b\} \in \wgdy(G^{\epsilon^*}, \sigma) \ \Big|\  A^a_{\ge k + 2}, A^b_{\le k + 1}\right)\pr_\sigma\left(A^a_{\ge k + 2}, A^b_{\le k + 1}\right) \\
		&= 2\pr_\sigma\left(\{a,b\} \in \wgdy(G^{\epsilon^*}, \sigma) \ \Big|\  A^a_{\le k + 1}, A^b_{\ge k + 2}\right)\pr_\sigma\left(A^a_{\le k + 1}, A^b_{\ge k + 2}\right) \\
		&= 2\pr_\sigma\left(\{a,b\} \in \wgdy(G^{\epsilon^*}, \sigma) \ \Big|\  A^a_{\le k + 1}, A^b_{\ge k + 2}\right)\pr_\sigma\left(A^b_{\ge k + 2} \ \Big|\  A^a_{\le k + 1}\right)\pr_\sigma\left(A^a_{\le k + 1}\right)\\
		&= 2\pr_\sigma\left(\{a,b\} \in \wgdy(G^{\epsilon^*}, \sigma) \ \Big|\  A^a_{\le k + 1}, A^b_{\ge k + 2}\right)\frac{k + 1}{2k + 1}\frac 12\\
		&= \frac{k + 1}{2k + 1}\pr_\sigma\left(\{a,b\} \in \wgdy(G^{\epsilon^*}, \sigma) \ \Big|\  A^a_{\le k + 1}, A^b_{\ge k + 2}\right) \text.
	\end{align*}
	
	Next, we sum over the number of agents $y \in Y$ that arrive in the second phase. 
	The probability that exactly $i$ agents from $Y$ arrive in the second phase can be computed using the hypergeometric distribution, i.e.,
	\begin{equation}\label{eq:hyperdist}
		\pr_\sigma\left(\big|\{y\in Y \colon A^y_{\ge k + 2}\}\big| = i \ \Big|\ A^a_{\ge k + 2}, A^b_{\le k + 1}\right) = \frac{\binom{k}{i}\binom{k}{k - i}}{\binom{2k}{k}}\text.
	\end{equation}
	
	We obtain
	\begin{align*}
		c_\wgdy
		&\le \frac{k + 1}{2k + 1} \sum_{i = 0}^{k} \pr_\sigma\left(\{a,b\} \in \wgdy(G^{\epsilon^*}, \sigma)  \ \Big|\ \big|\{y\in Y \colon A^y_{\ge k + 2}\}\big| = i, A^a_{\ge k + 2}, A^b_{\le k + 1} \right)\\
		&\phantom{=}\cdot \pr_\sigma\left(\big|\{y\in Y \colon A^y_{\ge k + 2}\}\big| = i \ \Big|\ A^a_{\ge k + 2}, A^b_{\le k + 1} \right) \\
		&\overset{\text{Eq.~(\ref{eq:hyperdist})}}{=} \frac{k + 1}{2k + 1} \sum_{i = 0}^{k} \frac{\binom{k}{i}\binom{k}{k - i}}{\binom{2k}{k}} \pr_\sigma\left(\{a,b\} \in \wgdy(G^{\epsilon^*}, \sigma) \ \Big|\  \big|\{y\in Y \colon A^y_{\ge k + 2}\}\big| = i, A^a_{\ge k + 2}, A^b_{\le k + 1}\right) \\
	\end{align*}
	
	Hence, we have now conditioned on $b$ as well as $i$ agents from $y \in Y$ arriving in the second phase, while $a$ arrives in the first phase.
	The probability that $\{a,b\}$ is formed is then equal to the probability that all agents that arrive before $b$ are from $Y$. 
	To compute this, we sum over the number of agents that have arrived before agent $b$ in the second phase. 
	If more than $i$ agents have arrived, then an agent in $X$ has arrived, and $\{a,b\}$ cannot form anymore.
	Hence, we have to consider the cases when $j$ assumes an integer between~$0$ and~$i$.
	We then calculate the probability that all $j$ agents arriving before $b$ are from $Y$. 
	Furthermore, the probability that $b$ arrives in a fixed position in the second half is $\frac{1}{k + 1}$. 
	We use \Cref{lem:inner} and \Cref{lem:outer}, whose rather techical proofs are deferred to \Cref{app:binom}.
	Once we simplify this expression, we see that $\wgdy\in \orderof{\frac{1}{n}}$.
	\begin{small}
	\begin{align*}
		c_\wgdy
		&= \frac{k + 1}{2k + 1} \sum_{i = 0}^{k} \frac{\binom{k}{i}\binom{k}{k - i}}{\binom{2k}{k}} \sum_{j = 0}^{i}  \pr_\sigma\left(\{a,b\} \in \wgdy(G^{\epsilon^*}, \sigma) \ \Big|\  A_{=k + 2 + j}^b, \big|\{y\in Y \colon A^y_{\ge k + 2}\}\big| = i, A^a_{\ge k + 2}, A^b_{\le k + 1}\right)\\
		&\phantom{=} \cdot \pr_\sigma\left(A_{=k + 2 + j}^b \ \Big|\  \big|\{y\in Y \colon A^y_{\ge k + 2}\}\big| = i, A^a_{\ge k + 2}, A^b_{\le k + 1}\right) \\
		&= \frac{k + 1}{2k + 1} \sum_{i = 0}^{k} \frac{\binom{k}{i}\binom{k}{k - i}}{\binom{2k}{k}} \sum_{j = 0}^{i}  \pr_\sigma\left(\{d \colon d \arr b, A^d_{\ge k + 2} \} \subseteq Y \ \Big|\   A_{=k + 2 + j}^b, \big|\{y\in Y \colon A^y_{\ge k + 2}\}\big| = i, A^a_{\ge k + 2}\right) \cdot \frac 1{k+1}\\ 
		&= \frac{k + 1}{2k + 1} \sum_{i = 0}^{k} \frac{\binom{k}{i}\binom{k}{k - i}}{\binom{2k}{k}} \sum_{j = 0}^{i}  \frac{\binom{i}{j}}{\binom{k}{j}}\frac{1}{k + 1} \overset{\text{\Cref{lem:inner}}}{=} \frac{k + 1}{2k + 1} \sum_{i = 0}^{k} \frac{\binom{k}{i}\binom{k}{k - i}}{\binom{2k}{k}} \frac{1}{k + 1}\frac{k + 1}{k + 1 - i}\\
		&= \frac{k + 1}{2k + 1} \sum_{i = 0}^{k} \frac{\binom{k}{i}\binom{k}{k - i}}{\binom{2k}{k}}\frac{1}{k + 1 - i} \overset{\text{\Cref{lem:outer}}}{=}  \frac{k + 1}{2k + 1}\frac{2k + 1}{(k + 1)^2}\\
		& = \frac{1}{k + 1} = \frac 2 n \in \orderof{\frac{1}{n}}
	\end{align*}
	\end{small}
	
	Hence, we have derived the desired bound.
\end{proof}

Next, we consider the performance of $\iew$.

\thmIEW*
	
\begin{proof}
	The lower bound of the competitive ratio follows directly from \Cref{thm:wgdy} and \Cref{lem:doubling} because $\iew = \mathit{I}$-$\wgdy$.
	
	For the upper bound of the competitive ratio of $\iew$, we once again consider the family of ASHGs given by $G^{k,\epsilon} = (N^{k,\epsilon},w^{k,\epsilon})$ in \Cref{fig:hard_instance}, as defined in the proof of \Cref{thm:compGDY}.
	Again, we analyze instances for fixed $k$ and omit $k$ as superscripts.
	Similar to the proofs of \Cref{thm:compGDY,thm:wgdy}, in the limit case for small $\epsilon$,  the optimal partition has a social welfare of $2$ and the social welfare of the partition computed by $\iew$ is equal to twice the probability of forming the coalition containing $a$ and $b$ for input $G^{\epsilon^*}$ with $\epsilon^* = \frac 12$.
	We only consider values of $n$ such that 
	\begin{equation}\label{eq:niceagentnumber}
		n = \sum_{i = 0}^{i^*} 2^{i+1}\text{ for some } i^* \in \mathbb N\text.
	\end{equation}
	
	This assures that $\iew$ completes all iterations.
	
	By definition, $\iew$ performs $\wgdy$ on $2^{i+1}$ agents in iteration $i$.
	The coalition $\{a, b\}$ can only form if $a$ and $b$ arrive in the same iteration~$i$.
	To compute the probability that this coalition forms in a given iteration~$i$, we proceed analogously to the proof in \Cref{thm:wgdy}.
	There is, however, one major difference that complicates the necessary computations.
	We cannot use the pigeonhole principle to determine whether agents from both $X$ and $Y$ exist in a phase, and we thus need to account for the exact number of agents from the sets $X$ and $Y$.
	As a consequence, it is also possible that the coalition $\{a, b\}$ forms even if both agents arrive in the second phase of the same iteration.
	Indeed, even if $a$ and $b$ both do not arrive in the first phase, it can still happen that only agents from $Y$ and $a$ arrive before $b$ arrives.
	
	Let us consider a fixed iteration $i$.
	There are $2^i$ agents that arrive in the first phase of iteration $i$ and $2^i$ agents that arrive in the second phase for a total of $2^{i+1}$ agents per iteration.
	The arrival time of the first agent from iteration $i$ is $1 + \sum_{j = 0}^{i - 1} 2^{j + 1} = 2^{i + 1} - 1$ and of the last agent is $\sum_{j = 0}^{i} 2^{j + 1} = 2^{i + 2} - 2$.
	To shorten the notation, let $A^x_{i, 1}$ and $A^x_{i, 2}$  denote the events that agent $x$ arrives in the first and second phase of iteration $i$, respectively.
	In other words, they represent the events $A^x_{\ge 2^{i + 1} - 1} \cap A^x_{\le 2^{i + 1} + 2^i - 2}$ and $A^x_{\ge 2^{i + 1} + 2^i - 1} \cap A^x_{\le 2^{i + 2} - 2}$, respectively.
	Summation over all $i^*$ phases yields
	\begin{align*}
		c_\iew & \le \pr_\sigma\left(\{a, b\} \in \iew\left(G^{\epsilon^*}, \sigma\right)\right) \\
		= & \sum_{i = 0}^{i^*} \left[\pr_\sigma\left(\{a, b\} \in \iew\left(G^{\epsilon^*}, \sigma\right) \Big|\  A^a_{i, 1}, A^b_{i, 2}\right) \pr_\sigma\left(A^a_{i, 1}, A^b_{i, 2}\right)\right.\\
		&\phantom{=} + \pr_\sigma\left(\{a, b\} \in \iew\left(G^{\epsilon^*}, \sigma\right) \Big|\  A^b_{i, 1}, A^a_{i, 2}\right) \pr_\sigma\left(A^b_{i, 1}, A^a_{i, 2}\right) \\
		&\phantom{=} \left.+ \pr_\sigma\left(\{a, b\} \in \iew\left(G^{\epsilon^*}, \sigma\right) \Big|\  A^a_{i, 2}, A^b_{i, 2}\right) \pr_\sigma\left(A^a_{i, 2}, A^b_{i, 2}\right)\right] \\
	\end{align*}
	
	Furthermore, we estimate the case where both the agents $a$ and $b$ arrive in the second phase, with the case where $a$ arrives in the first and $b$ in the second phase.
	For this, we can define an injective mapping of arrival orders where $a$ and $b$ arrive in the second phase by swapping the earlier one arriving as $j$th agent in the second phase with the $j$th agent in the first phase.
	This maps to arrival orders where $a$ and $b$ arrive in different phases of the same iteration.
	Moreover, whenever the coalition $\{a,b\}$ while $a$ and $b$ arrive in the second phase, then all agents before and between $a$ and $b$ are from one of the two sets $X$ and $Y$.
	Hence, $\{a,b\}$ still forms in the mapped instance.
	This yields
	\begin{equation*}% 
		 \begin{aligned}
		 &\pr_\sigma\left(\{a, b\} \in \iew\left(G^{\epsilon^*}, \sigma\right) \Big|\  A^a_{i, 2}, A^b_{i, 2}\right) \pr_\sigma\left(A^a_{i, 2}, A^b_{i, 2}\right) \\
		 &\phantom{=}\le \pr_\sigma\left(\{a, b\} \in \iew\left(G^{\epsilon^*}, \sigma\right) \Big|\  A^a_{i, 1}, A^b_{i, 2}\right) \pr_\sigma\left(A^a_{i, 1}, A^b_{i, 2}\right)\\
		 &\phantom{=} + \pr_\sigma\left(\{a, b\} \in \iew\left(G^{\epsilon^*}, \sigma\right) \Big|\  A^b_{i, 1}, A^a_{i, 2}\right) \pr_\sigma\left(A^b_{i, 1}, A^a_{i, 2}\right) 
		 \end{aligned}
	\end{equation*}
	
	We use this insight together with symmetry to combine the cases where $a$ arrives before and after~$b$ to one case where $a$ arrives before $b$.
	We obtain
	\begin{align}
		c_\iew 
		& \le \sum_{i = 0}^{i^*} 2\left[\pr_\sigma\left(\{a, b\} \in \iew\left(G^{\epsilon^*}, \sigma\right) \Big|\  A^a_{i, 1}, A^b_{i, 2}\right) \pr_\sigma\left(A^a_{i, 1}, A^b_{i, 2}\right)\right.\notag\\
		& + \left.\pr_\sigma\left(\{a, b\} \in \iew\left(G^{\epsilon^*}, \sigma\right) \Big|\  A^b_{i, 1}, A^a_{i, 2}\right) \pr_\sigma\left(A^b_{i, 1}, A^a_{i, 2}\right) \right]\notag\\ 
		& \le \sum_{i = 0}^{i^*} 4\pr_\sigma\left(\{a, b\} \in \iew\left(G^{\epsilon^*}, \sigma\right) \Big|\  A^a_{i, 1}, A^b_{i, 2}\right) \pr_\sigma\left(A^a_{i, 1}, A^b_{i, 2}\right) \notag\\
		& =  \sum_{i = 0}^{i^*} 4\pr_\sigma\left(\{a, b\} \in \iew\left(G^{\epsilon^*}, \sigma\right) \Big|\  A^a_{i, 1}, A^b_{i, 2}\right) \frac{2^i}{n}\frac{2^i}{n-1} \notag\\
		& = \sum_{i = 0}^{i^*} \frac{2^{2i + 2}}{n^2 - n}\pr_\sigma\left(\{a, b\} \in \iew\left(G^{\epsilon^*}, \sigma\right) \Big|\  A^a_{i, 1}, A^b_{i, 2}\right) \text.\label{eq:initial}
	\end{align}
	
	Next, we compute $\pr_\sigma\left(\{a, b\} \in \iew\left(G^{\epsilon^*}, \sigma\right) \big|\  A^a_{i, 1}, A^b_{i, 2}\right)$ for iteration $i$. 
	First, we sum over the number of agents $j$ from $Y$ that arrive in the second phase of iteration $i$, and then we sum over all relevant arrival times of $b$.
	These are in the range $[2^{i+1} + 2^i - 1, 2^{i+1} + 2^i - 1  + j]$ as in all other cases the probability is $0$.
	Next, we compute the probability that the coalition $\{a, b\}$ forms using the hypergeometric distribution as we did in \Cref{thm:wgdy}.
	We then simplify the expression and finally bound it from above by using the inequality derived in \Cref{lem:crazy} in \Cref{app:binom}, which holds for $i, k \in \mathbb{N}_0$ and $k + 2 > 2^i$.
	Note that these conditions are fulfilled since $i$ and $k$ are non-negative integers and $2k + 2 = n = \sum_{i = 0}^{i^*} 2^{i + 1} = 2\sum_{i = 0}^{i^*} 2^{i}$ which implies $\sum_{i = 0}^{i^*} 2^{i} = k + 1 < k + 2$.
	\begin{align*}
		& \pr_\sigma\left(\{a, b\} \in \iew\left(G^{\epsilon^*}, \sigma\right) \Big|\  A^a_{i, 1}, A^b_{i, 2}\right) \\
		= & \sum_{j = 0}^{2^i - 1}\pr_\sigma\left(\{a, b\} \in \iew\left(G^{\epsilon^*}, \sigma\right) \Big|\ \big|\{y \in Y: A^y_{i, 2}\}\big| = j, A^a_{i, 1}, A^b_{i, 2} \right)\cdot \pr_\sigma \left(\big|\{y \in Y: A^y_{i, 2}\}\big| = j \Big|\,A^a_{i, 1}, A^b_{i, 2}\right) \\
		= & \sum_{j = 0}^{2^i - 1} \frac{\binom{k}{j}\binom{k}{2^i - 1 - j}}{\binom{2k}{2^i-1}} \pr_\sigma\left(\{a, b\} \in \iew\left(G^{\epsilon^*}, \sigma\right) \Big|\  \big|\{y \in Y: A^y_{i, 2}\}\big| = j, A^a_{i, 1}, A^b_{i, 2} \right) \\
		= & \sum_{j = 0}^{2^i - 1} \frac{\binom{k}{j}\binom{k}{2^i - 1 - j}}{\binom{2k}{2^i-1}} \sum_{k = 0}^{j} \pr_\sigma\left(\{a, b\} \in \iew\left(G^{\epsilon^*}, \sigma\right) \Big|\  A^b_{= 2^{i + 1} + 2^i - 1 + k}, \big|\{y \in Y: A^y_{i, 2}\}\big| = j, A^a_{i, 1}, A^b_{i, 2}\right)\\
		&\phantom{=}\cdot \pr_\sigma\left(A^b_{= 2^{i + 1} + 2^i - 1 + k} \Big|\  \big|\{y \in Y: A^y_{i, 2}\}\big| = j, A^a_{i, 1}, A^b_{i, 2}\right) \\
		= & \sum_{j = 0}^{2^i - 1} \frac{\binom{k}{j}\binom{k}{2^i - 1 - j}}{\binom{2k}{2^i-1}} \sum_{k = 0}^{j} \frac{1}{2^i} \pr_\sigma\left(\{a, b\} \in \iew\left(G^{\epsilon^*}, \sigma\right) \Big|\  A^b_{= 2^{i + 1} + 2^i - 1 + k}, \big|\{y \in Y: A^y_{i, 2}\}\big| = j, A^a_{i, 1} \right) \\
		= & \sum_{j = 0}^{2^i - 1} \frac{\binom{k}{j}\binom{k}{2^i - 1 - j}}{\binom{2k}{2^i-1}} \sum_{k = 0}^{j} \frac{1}{2^i} \pr_\sigma\left(\{d \in N \setminus \{a, b\}: d \succ^\sigma b \land A^d_{i, 2}\} \subseteq Y \Big|\  A^b_{= 2^{i + 1} + 2^i - 1 + k}, \big|\{y \in Y: A^y_{i, 2}\}\big| = j, A^a_{i, 1} \right) \\
		= & \sum_{j = 0}^{2^i - 1} \frac{\binom{k}{j}\binom{k}{2^i - 1 - j}}{\binom{2k}{2^i-1}} \sum_{k = 0}^{j} \frac{1}{2^i} \frac{\binom{j}{k}}{\binom{2^i-1}{k}} 
		= \sum_{j = 0}^{2^i - 1} \frac{\binom{k}{j}\binom{k}{2^i - 1 - j}}{\binom{2k}{2^i-1}} \frac{1}{2^i - j}  \overset{\text{\Cref{lem:crazy}}}{\le} \frac{2}{2^i} \\
	\end{align*}
	
	Inserting this to \Cref{eq:initial} and applying \Cref{eq:niceagentnumber} yields
	\begin{align*}
		c_\iew \le & \sum_{i = 0}^{i^*} \frac{2^{2i + 2}}{n^2 - n}\pr_\sigma\left(\{a, b\} \in \iew\left(G^{\epsilon^*}, \sigma\right) \Big|\  A^a_{i, 1}, A^b_{i, 2}\right)  \\
		\le & \sum_{i = 0}^{i^*} \frac{2^{2i + 2}}{n^2 - n} \frac{2}{2^i} 
		=   \frac{4}{n^2 - n} \sum_{i = 0}^{i^*} 2^{i + 1} 
		=  \frac{4n}{n^2 - n} = \frac{4}{n - 1} \in \orderof{\frac{1}{n}} \text.
	\end{align*}
	
	We have shown that the competitive ratio is $\orderof{\frac{1}{n}}$. We conclude that $\iew$ has a competitive ratio of $c_\iew \in \Theta\left(\frac{1}{n}\right)$ and this bound is tight.
\end{proof}

\section{Optimal Stopping Algorithm}\label{ap:odds}

This section aims to prove \Cref{prop:maxe}.
Throughout most of the section, we assume that the number of arriving agents is known to the algorithm. 
Ultimately, we can use \Cref{lem:doubling} to generalize the method to unknown $n$ while maintaining a constant competitive ratio.\footnote{This does, however, not mean that we can use our method as an optimal stopping algorithm for unknown $n$ because the doubling variant of our algorithm stops once in every phase.
Indeed, it is unclear how to extend optimal stopping algorithms to unknown~$n$ \cite[Chapter~4]{Brus00a}.}
We introduce a new algorithm that we will analyze under a random arrival order.

The key idea is to interpret the computation of a partition in an online ASHG as an optimal stopping problem for independent events.
Let $J_1, \dots, J_n$ be mutually independent indicator events.
A time step $k$ is called a \emph{success time} if $J_k = 1$.
In an \emph{optimal stopping problem}, the indicator events $J_1, \dots, J_n$ are observed sequentially, and in each step, the algorithm may stop the process.
The goal is to stop at the last success time, i.e., the goal is to maximize the probability of stopping at the last time step where the corresponding event happens.

An optimal algorithm for this problem is the odds strategy that stops at the first success time after a certain stopping time.
The optimal stopping time $s$, as well as the resulting success probability of the algorithm, can be computed in sublinear time \cite[Theorem~1]{Brus00a}.
Let $p_k = \pr(J_k = 1)$ be the probability that the $k$th event happens and let $r_k = \frac{p_k}{1 - p_k}$ be the so called \emph{odds} of $I_k$.
The odds algorithm then sums the odds in reverse order until their sum exceeds one at time step $s$ or sets $s = 1$ otherwise.
The algorithm then returns the optimal stopping time $s$.
More precisely, if we are in the first case, the algorithm returns the largest $s$ such that $\sum_{i = s}^{n} r_i \ge 1$.

For $2\le k \le n$, let $J_k$ be the event 
\begin{quote}
	``the maximum weight edge connected to the $k$th agent is strictly larger than all edges among the first $k-1$ agents.''
\end{quote}

We show next that the events $J_2, \dots, J_n$ are mutually independent. 
This allows us to execute the odds algorithm.

We start with some notation.
Let $E_k := \left\{\{a, b\} \in {N\choose 2} \colon A^a_{\le k} \land A^b_{\le k}\right\}$ and $e_k := \arg \max_{e \in E_k} w(e)$ where we assume $e_k = \left(a_k, b_k\right)$.
The set $E_k$ contains all edges between the first $k$ agents, edge $e_k$ is the maximal weight edge among those, and $a_k$ and $b_k$ are the agents that are connected by $e_k$.
It holds that event $J_k$ occurs if and only if $A^{a_k}_{= k}$ or $A^{b_k}_{= k}$.

\begin{lemma}\label{lem:oddsindep}
	Let an ASHG be given for which all edge weights are pairwise different.
	Consider indices $2 \le k_1 < k_2 < \dots < k_j \le n$ for some $j\le n$. Then, it holds that
	\begin{equation}\label{eq:independence_lhs}
		\pr\left(\bigcap_{i = 1}^j J_{k_i}\right) = \prod_{i = 1}^{j} \frac{2}{k_i} \text.
	\end{equation}
	In particular, the events $J_2, \dots, J_n$ are mutually independent.
\end{lemma}

\begin{proof}
	Let $2 \le k_1 < k_2 < \dots < k_j \le n$ for $j \le n$ be the arrival times at which event $J_{k_i}$ must happen and let $1 \le \ell_1 < \ell_2 < \dots < \ell_{n - j} \le n$ be all other arrival times.
	We first count the number of arrival orders where the events $J_{k_1}, J_{k_2}, \dots, J_{k_j}$ happen simultaneously.
	In particular, it is irrelevant if the events $J_{\ell_1}, \dots, J_{\ell_{n - j}}$ happen or not.
	We consider the agents in the reverse arrival order.
	Then, if arrival time $t$ is contained in $\{k_1, \dots, k_j\}$, we have just two possible agents that can arrive at time $t$, i.e., the two agents that currently have the maximum weight edge between them.
	Otherwise, we have $t$ choices because we can choose among all $t$ present agents.
	Together, this shows that the total number of arrival orders in which the events $J_{k_1}, J_{k_2}, \dots, J_{k_j}$ happen simultaneously is $\prod_{i = 1}^{j}2\prod_{i=1}^{n-j}\ell_i$.
	Next, we divide by the total number of arrival orders $n!$ and get
	
	\begin{equation*}
		\pr\left(\bigcap_{i = 1}^j J_{k_i}\right) = \frac{\prod_{i = 1}^{j}2\prod_{i=1}^{n-j}\ell_i}{n!} = \frac{\prod_{i = 1}^{j}2}{\prod_{i = 1}^{j}k_i} = \prod_{i = 1}^{j} \frac{2}{k_i} \text.
	\end{equation*}
	
	This proves \Cref{eq:independence_lhs}.
	As a consequence,
	\begin{equation*}
		\pr\left(\bigcap_{i = 1}^j J_{k_i}\right) = \prod_{i = 1}^{j} \frac{2}{k_i} = \prod_{i = 1}^{j} \pr(J_{k_i})\text.
	\end{equation*}	
	The second equality follows from applying \Cref{eq:independence_lhs} for single events.
	Hence, we have shown mutual independence of $J_2, \dots, J_n$.
\end{proof}

We consider the following algorithm.
\begin{definition}[Maximum edge algorithm]
	The \emph{maximum edge algorithm} ($\maxe$) executes the odds algorithm \citep{Brus00a} offline to compute an optimal stopping time $s$, where $n$ and $p_k = \frac{2}{k}$ for all $k \in \{2, \dots, n\}$ is used as input. Next, the odds strategy is performed online, i.e., no coalition is formed until at least $s$ agents arrived, then the first edge that is larger than all previously seen edges is matched if its weight is strictly larger than $0$.
	If multiple such edges arrive in a step, then the one with highest weight is matched.
\end{definition}

The key insight of identifying $\maxe$ with an optimal stopping problem is captured in the following lemma.

\begin{lemma}\label{lem:stopping}
	If the odds algorithm for input length $n$ and probabilities $p_k = \frac{2}{k}$ for all $k \in \{2, \dots, n\}$ stops with the last $1$, then $\maxe$ outputs a partition of social welfare $2 w(e_{\max})$.
\end{lemma}

\begin{proof}
	The arrival of the largest edge corresponds to the last time $k$ where $J_k$ occurs.
\end{proof}

We now show that $\maxe$ has a constant competitive ratio on our restricted domain inspired by our worst-case family.
We show the statement first assuming a fixed number of agents, and then apply \Cref{lem:doubling}.

\PropMaxe*

\begin{proof}
	Consider first the case of an ASHG given by a graph $G$ with pairwise different edge weights.
	By \Cref{lem:oddsindep}, we can apply the odds theorem \cite[Theorem~1]{Brus00a} for input length $n$ and probabilities $p_k = \frac{2}{k}$ for all $k \in \{2, \dots, n\}$. 
	Then, the odds algorithm has a probability of at least $\frac{1}{2e}$ to terminate successfully \cite[Theorem~2]{Brus00a}.
	There, $e$ denotes Euler's number.
	Moreover, by \Cref{lem:stopping}, $\maxe$ computes a partition of social welfare $2 w(e_{\max})$ if the odds algorithm terminates successfully.
	Hence, since $w(e_{\max}) \ge \lambda \cdot\SW[\pi^*(G)]$, we get that $\EV_\sigma[\SW[\maxe(G,\sigma)]] \ge \frac{\lambda}{e} \SW[\pi^*(G)]$.% 
	
	Now, let us consider games in which edge weights are not pairwise different.
	Note that $\maxe$ still can be used for these instances.
	We will now obtain a performance guarantee for a game by transforming it to a related game and using the performance guarantee of $\maxe$ on the related game.
	
	Therefore, now consider an arbitrary game $G = (N,E,w)$.
	We transform it into a related game $G' = (N,E,w')$ as follows. 
	First, define $\epsilon := \frac 12 \min\{|w(e)- w(e')|\colon e,e'\in E, w(e) \neq w(e')\}$.
	This is half of the minimum weight difference of any pair of edges.
	Now, sort edges in an arbitrary order, i.e., $E = \{e_1,\dots, e_m\}$ where $m = \frac {n(n-1)}2$.
	For $\ell\in [m]$, we iteratively define $w'(e_\ell)$ as follows:
	we choose the smallest $x\in \{\frac i{n^2} \epsilon \colon i = 0,1,\dots, n^2\}$ such that $w(e_{\ell})-x \neq w'(e_j)$ for all $j \in [\ell -1]$ and define $w'(e_{\ell}) = w(e_{\ell}) - x$.
	Note that we choose among $n^2+1\ge m$ values, and hence, such an $x$ always exists.
	Hence, the weight function $w'$ copies the value of $w$ if it is a new value, and lowers the value of the edge otherwise by a small amount.
	This amount is chosen in a way that essentially breaks weight ties among same weight edges while leaving the relative order of edge weights intact.
	In other words, for all $e,e'\in E$, we have that $w'(e) > w'(e')$ whenever $w(e) > w(e')$.

	Since the procedure to define $w'$ only modifies the weights of edges whose weight already occurs, the weight of the maximum weight edge in $G$ and $G'$ is identical and equal to $w(e_{\max})$.
	Moreover, we observe that in $G'$ the maximum weight edge is unique and all weights are different.
	Hence, by the first part of the proof, $\maxe$ terminates successfully with probability at least $\frac 1{2e}$ on $G'$.
	
	We now claim that $\maxe$ terminates successfully on $G$ for all arrival orders for which it terminates successfully on $G'$.
	To show this, consider a fixed arrival oder for which $\maxe$ terminates successfully on $G'$. 
	Assume for contradiction that $\maxe$ forms a coalition on $G$ at the arrival of an earlier agent. 
	Hence, when this agent arrives, it is part of an edge $e$ such that $w(e) > w(e')$ for all previously encountered edges $e'$.
	Since $w'(e) \ge w(e) -\frac 12 \epsilon$, we have that $w'(e) > w(e')$ for all previously encountered edges $e'$, since the choice of $\epsilon$ ensures that this edge still has a weight higher than the weight of the next lower weight.
	Moreover, since $w'$ only lowers edge weights, we have that $w(e') \ge w'(e')$, it holds that $w'(e) > w'(e')$ for all previously encountered edges $e'$.
	Hence, $\maxe$ has to match $e$, but this is no success even as $e_{\max}$ only arrives later.
	We obtain a contradiction.
	
	Hence, $\maxe$ cannot form a coalition on $G$ until the arrival of $e_{\max}$.
	However, at the arrival of $e_{\max}$, this edge has a weight strictly larger than the weight of all previously arrived edges (as it is the unique edge of maximum weight in both $G$ and $G'$).
	Hence, $\maxe$ forms a coalition with this edge, and terminates with a success event.
	Since the chosen arrival order was arbitrary, we conclude that the probability that $\maxe$ terminates successfully on $G$ is at least as high as the probability that $\maxe$ terminates successfully on $G'$, i.e., it is at least $\frac 1{2e}$.
	Hence, since $w(e_{\max}) \ge \lambda \cdot\SW[\pi^*(G)]$, we get that $\EV_\sigma[\SW[\maxe(G,\sigma)]] \ge \frac{\lambda}{e} \SW[\pi^*(G)]$.

	Hence, it holds that $c_{\maxe} \ge \frac{\lambda}{e}\in \Theta\left(1\right)$.
	By \Cref{lem:doubling}, $\mathit{I}$-$\maxe$ has the desired competitive ratio. 
\end{proof}

\section{Online Matching under Free Dissolution}

In this section, we provide the proofs for our general model of online matching under free dissolution.
As we mentioned before, the analysis of $\left(1 + \frac {\sqrt{2}}2\right)$-$\ggdy$ only requires a straightforward adaptation of the techniques by \citet{FKM+05a} and \citet{McGr05a}.

\THMggdy*

\begin{proof}
	For a simpler exposition, we write $\ggdy$ instead of $\left(1 + \frac {\sqrt{2}}2\right)$-$\ggdy$ in this proof.
	We show first that $c_\ggdy\le \frac{1}{3 + 2\sqrt{2}}$. To this end, we define a family of instances $(G^{k,\epsilon})_{k\ge 1,\epsilon > 0}$ with a sufficiently large gap between the social welfare of the algorithmic and optimal partitions.
	The construction is depicted in \Cref{fig:ggdy}. 
	\begin{figure}[tbp]
		\centering
		\resizebox{1\textwidth}{!}{
			\begin{tikzpicture}
				\foreach \i/\j/\k in {0/1/0,1/2/1,2/3/2,3/4/3,4/{k+1}/k,5/{k+1}/{k+1}}
				{
					\node[draw, circle](a\i) at (4*\i,0){\color{white}{$a33$}};
					\node[draw, circle](b\i) at (4*\i,3){\color{white}{$a33$}};
					\node at (4*\i,0){$a_{\k}$};
					\node at (4*\i,3){$b_{\k}$};
					\draw (a\i) edge node[pos = 0.5, fill = white]{$\left(1 + \frac{\sqrt{2}}{2}\right)^{\j}-\epsilon$} (b\i);
				}
				\draw (a0) edge node[pos = 0.5, fill = white]{$1$} (a1);
				\draw (a1) edge node[pos = 0.5, fill = white]{$1 + \frac{\sqrt{2}}{2}$} (a2);
				\draw (a2) edge node[pos = 0.5, fill = white]{$\left(1 +\frac{\sqrt{2}}{2}\right)^2$} (a3);
				\draw (a4) edge node[pos = 0.5, fill = white]{$\left(1 + \frac{\sqrt{2}}{2}\right)^{k}$} (a5);
				\node at (barycentric cs:a3=1,a4=1) {\dots};
			\end{tikzpicture}
		} 
		\caption{Family of instances for tightness of competitive ratio in \Cref{thm:ggdy}.\label{fig:ggdy}}
	\end{figure}
	
	Let $G^{k,\epsilon} = (N^{k,\epsilon},w^{k,\epsilon})$ where $N^{k,\epsilon} = \{a_i,b_i\colon 0\le i\le k+1\}$.
	The edge weights are given by $w^{k,\epsilon}(a_i,b_i) = \left(1 + \frac{\sqrt{2}}{2}\right)^{i+1}-\epsilon$ for $0\le i\le k$, $w^{k,\epsilon}(a_{k+1},b_{k+1}) = \left(1 + \frac{\sqrt{2}}{2}\right)^{k+1}-\epsilon$, and $w^{k,\epsilon}(a_i,a_{i+1}) = \left(1 + \frac{\sqrt{2}}{2}\right)^i$ for $0\le i\le k$.
	All other weights are set to $0$.
	The arrival order for $G^{k,\epsilon}$ is independent of $\epsilon$ and is given by $\sigma^k = (a_0,a_1,b_0,a_2,b_1,\dots,a_{k+1},b_k,b_{k+1})$.
	
	Consider an execution of $\ggdy$ for $(G^{k,\epsilon},\sigma^k)$.
	It is easy to see that the only non-singleton coalition contained in $\ggdy(G^{k,\epsilon},\sigma^k)$ is $\{a_k,a_{k+1}\}$, hence it achieves a weight of 
	$\left(1 + \frac{\sqrt{2}}{2}\right)^k$.
	On the other hand, the maximum weight matching for $G^{k,\epsilon}$ (and sufficiently small $\epsilon$) is $\genmat^*(G^{k,\epsilon}) = \{\{a_i,b_i\}\colon 0\le i\le k+1\}$ with a weight of
	$\sum_{i = 0}^{k + 1} \left(1 + \frac{\sqrt{2}}{2}\right)^i + \left(1 + \frac{\sqrt{2}}{2}\right)^{k + 1} -(k+2)\epsilon$.
	Hence,	
	\begin{align*}
		c_\ggdy & = \inf_{G,\sigma}\frac{\SW[\ggdy(G,\sigma)]}{\SW[\genmat^*(G,\sigma)]} \le \inf_{k\ge 1,\epsilon > 0}\frac{\SW[\ggdy(G^{k,\epsilon},\sigma^k)]}{\SW[\genmat^*(G^{k,\epsilon})]}\\
		& = \inf_{k\ge 1,\epsilon > 0}\frac{\left(1 + \frac{\sqrt{2}}{2}\right)^k}{\sum_{i = 0}^{k + 1} \left(1 + \frac{\sqrt{2}}{2}\right)^i + \left(1 + \frac{\sqrt{2}}{2}\right)^{k + 1} -(k+2)\epsilon} \\
		& = \inf_{k \ge 1} \frac{1}{\sum_{i = 0}^{k - 1} \frac{\left(1 + \frac{\sqrt{2}}{2}\right)^i}{\left(1 + \frac{\sqrt{2}}{2}\right)^k} + 1 + 2\left(1 + \frac{\sqrt{2}}{2}\right)} 
		 = \inf_{k \ge 1} \frac{1}{\sum_{i = 1}^{k} \left(1 + \frac{\sqrt{2}}{2}\right)^{-i} +3 + \sqrt{2}} 
	\end{align*}
	
	To simplify further, we show that 
	\begin{equation}\label{eq:partial_sum_bound}
		\sum_{i = 1}^{k} \left(1 + \frac{\sqrt{2}}{2}\right)^{-i} = \sqrt{2}\left(1 - \left(1 + \frac{\sqrt{2}}{2}\right)^{-k}\right)\text.
	\end{equation}

	To this end, we use the standard formula for the $k$th partial sum of a geometric series: For all $x\neq 1$, it holds that $\sum_{i = 1}^k x^i = x\frac{1-x^k}{1-x}$.
	Hence, we have to show that
	$\frac {x^*}{1-x^*} = \sqrt{2}$ for $x^* =  \left(1 + \frac{\sqrt{2}}{2}\right)^{-1}$.
	Note that 
	\begin{equation*}
		x^* = \left(1 + \frac{\sqrt{2}}{2}\right)^{-1} =  \left(\frac{2+ \sqrt{2}}{2}\right)^{-1} =  \left(\frac{2}{2(2-\sqrt{2})}\right)^{-1} = 2 - \sqrt{2}\text.
	\end{equation*}
	
		Thus, 
		\begin{equation*}
			\frac{x^*}{1 - x^*} = \frac{2 - \sqrt{2}}{\sqrt{2} - 1} = \frac{\sqrt{2}\left(\sqrt{2} - 1\right)}{\sqrt{2} - 1} = \sqrt{2}\text.
		\end{equation*}
		
	This establishes \Cref{eq:partial_sum_bound}.
	Returning to our estimation of the competitive ratio, we get
	
	\begin{equation*}
		c_\ggdy \le \inf_{k \ge 1} \frac{1}{\sum_{i = 1}^{k} \left(1 + \frac{\sqrt{2}}{2}\right)^{-i} +3 + \sqrt{2}} 
		= \inf_{k \ge 1} \frac{1}{\sqrt{2}\left(1 - \left(1 + \frac{\sqrt{2}}{2}\right)^{-k}\right) +3 + \sqrt{2}} 
		= \frac{1}{3 + 2\sqrt{2}}\text.
	\end{equation*}
	
	Next, we show that $c_\ggdy \ge  \frac{1}{3 + 2\sqrt{2}}$.
	Let $\genmat$ be the matching produced by $\ggdy$ and let $\genmat^*$ be a maximum weight matching.
	Without loss of generality, we may assume that all edges in $\genmat^*$ have positive weight.
	Further, let $F \subseteq E$ be the set of edges that were formed by $\ggdy$ at some point, i.e., $F$ consists of the edges in $\genmat$ as well as all edges that have been dissolved.
	The key idea is to consider the relation induced by the replacement of edges.
	More precisely, given two edges $e,e'\in F$, we say that $e$ \emph{dominates} $e'$ with respect to replacement, written as $e \succ_R e'$, if there exists a chain of edges $e_0,\dots, e_j$ such that $e_0 = e'$, $e_j = e$ and for $i\in [j]$, the formation of $e_i$ has lead to the dissolution of $e_{i-1}$.
	Additionally, for each edge $e \in F$, we set $e \sim_R e$ and, for all $e, e' \in F$, we define $e \succsim_R e'$ if $e \succ e'$ or $e \sim e'$.
	Note that $\genmat$ consists precisely of the maximal elements in $F$ with respect to $\succ_R$.
	For a set of edges $F'\subseteq F$, we define $\max_{\succ_R}(F') = \{f\in F'\colon \nexists f'\in F \text{ with }f'\succ_R f\}$.
	In particular, every edge that is matched by the algorithm $e \in \genmat $ is the maximal element of an inclusion-maximal, transitive subset of $F$.

	We now establish a connection between $\genmat^*$ and $\genmat$ so that we can bound the weight of the former by the latter.
	We do this in three steps, each shown in a claim below.
	First, we define a function $\matfn$ that maps each edge $m  = \{a,b\}\in \genmat^*$ to one or two edges in $F$.
	Intuitively, $\matfn$ follows the trace of edge replacing the edges matched by $a$ and $b$ at the arrival of the second agent in $m$.
	It then maps to the maximal edges with respect to domination that still contain $a$ or $b$.
	We will illustrate this function in an example after its formal definition.
	The first step is to show that $m$ maps to edges that have a weight that not much worse than its own weight.
	More precisely, in \Cref{cl:weightcorespondence}, we show that the weight of $m$ is bounded by $(1 + \frac{\sqrt{2}}{2})$ times the weight of the edges in $\matfn(m)$.
	Hence, to bound the weight of $\genmat^*$, we can bound the weight of edges in $F$.
	This happens in the next two claims. 
	In \Cref{cl:fewmatches}, we show that each edge $e \in \genmat$ (i.e., edges produced by the algorithm) occurs in the image of $\matfn$ at most twice, and each edge in $F \setminus \genmat$ at most once.\footnote{The latter fact is the reason why we do not simply let $\rho$ map to the exact edge(s) in $F$ that are matched instead of $m$. If we did this, then an edge in $F \setminus \genmat$ could be mapped to twice.}
	Hence, it essentially remains to bound the weight of replaced edges.
	Therefore, in \Cref{cl:weightchains}, we show that for each $e \in \genmat$, the edges in $F$ dominated by $e$ have a total weight of at most $\sqrt{2}w(e)$.
	Combining all insights achieves a bound on the weight of $\genmat^*$ in terms of the weight of $\mu$, which yields the desired bound on the competitive ratio.
	
	To formalize this argument, we now define a function $\matfn\colon \genmat^*\to 2^F$ that captures the idea of mapping each edge $m \in \genmat^*$ to edges in $F$.
	Let $m =\{a,b\}\in \genmat^*$ and assume that $b$ arrives after $a$.
	The definition of $\matfn$ depends on whether and how $a$ and $b$ get matched at the arrival of $b$, and is split into three cases.
	If, at the arrival of $b$, $a$ is already matched with $c$ and $\left(1 + \frac{\sqrt{2}}{2}\right) w(a,c) > w(a,b)$, then $\matfn(m) = \max_{\succ_R}\{f\in F\colon f\succsim_R \{a,c\}, a\in f\}$.
	If, at the arrival of $b$, $a$ is already matched with $c$ and $\left(1 + \frac{\sqrt{2}}{2}\right) w(a,c) \le w(a,b)$, then $b$ is matched with some $d$ and we define $\matfn(m) = \max_{\succ_R}(\{f\in F\colon f\succsim_R \{a,c\}, a\in f\}\cup\{f\in F\colon f\succsim_R \{b,d\}, b\in f\})$.
	Note that, in this case, it can happen that $d = a$, i.e., $a$ is matched with $b$. 
	In this case, we have that $\{b,d\}\succ_R \{a,c\}$, and we only map to one edge, namely the last edge in the chain of replacements of $\{a,b\}$ that contains either $a$ or $b$.
	If at the arrival of $b$, $a$ is unmatched, then $b$ is matched with some $d$ (since $w(m)>0$, matching with $a$ is a valid option) and we define $\matfn(m) = \max_{\succ_R}\{f\in F\colon f\succsim_R \{b,d\}, b\in f\}$.
	It can also happen in this case that $d = a$, i.e., $a$ is matched with $b$.
	Note that $\matfn$ always maps to a non-empty set.
	Moreover, we map to a maximal element in $F$ that contains $a$ or $b$ as this is needed in \Cref{cl:fewmatches}.
	Indeed, 
	if $a$ is unmatched or matched to an agent~$c$ with $\left(1 + \frac{\sqrt{2}}{2}\right) w(a,c) \le w(a,b)$, then $b$ will certainly be matched because matching with $a$ is an eligible option.
	Note that $\matfn$ is a set-valued function, but for all $m\in \genmat^*$, it holds that $|\matfn(m)|\le 2$.
	Moreover, the only case where $\matfn(m)$ contains two edges is if we are in the second case and the edge $\{a,b\}$ is not created.

		\begin{figure}[tbp]
		\centering
			\begin{tikzpicture}
				\pgfmathsetmacro\figscale{2.5}
				\foreach \i/\j/\k in {0/0/c,\figscale/0/a,2*\figscale/0/b,3*\figscale/0/d}
				{
					\node[draw, circle] (\k) at (\i,\j){\color{white}{$a33$}};
					\node at (\k) {$\k$};
				}
				\node[draw, circle] (r) at ($(a) + (225:\figscale)$) {\color{white}{$a33$}};
					\node at (r) {$r$};
				\node[draw, circle] (s) at ($(r) + (0:\figscale)$) {\color{white}{$a33$}};
					\node at (s) {$s$};
				\node[draw, circle] (t) at ($(d) + (225:\figscale)$) {\color{white}{$a33$}};
					\node at (t) {$t$};
					
					\draw (c) edge[very thick, TUMBlue] node[pos = 0.5, fill = white]{\color{black}{$2 / 1$}} (a);
					\draw (a) edge node[pos = 0.5, fill = white]{$2$} (b);
					\draw (a) edge node[pos = 0.5, fill = white]{$4$} (r);
					\draw (r) edge[very thick] node[pos = 0.5, fill = white]{$8$} (s);
					\draw (b) edge[very thick, TUMBlue] node[pos = 0.5, fill = white]{\color{black}{$4$}} (d);
					\draw (d) edge[very thick] node[pos = 0.5, fill = white]{$8$} (t);
					
			\end{tikzpicture}
		\caption{Illustration of definition of $\matfn$. 
			Consider the arrival order $c,a,d,b,r,s,t$. 
			The bold black edges indicate the matching returned by $\ggdy$.
			The blue edges indicate the matching that will be matched tentatively at the arrival of $b$.
			If $w(a,c) = 2$, then we have $\matfn(\{a,b\}) = \{\{a,r\}\}$.
			If $w(a,c) = 1$, then we have $\matfn(\{a,b\}) = \{\{a,r\},\{b,d\}\}$. \label{fig:rho}}
	\end{figure}
	
	We illustrate the definition of $\matfn$ in \Cref{fig:rho}.
	We only depict edges of non-zero weight and assume that the arrival order is $c,a,d,b,r,s,t$.
	The weight of $\{a,c\}$ can be $2$ or $1$.
	At the arrival of $a$, $\ggdy$ forms $\{a,c\}$.
	Then, $d$ remains unmatched at her arrival.
	When $b$ arrives, independently of the weight of  $\{a,c\}$, $\ggdy$ forms $\{b,d\}$.
	Finally, when $r$, $s$, and $t$ arrive, we each time dissolve an edge to form $\{a,r\}$, $\{r,s\}$, and $\{d,t\}$, respectively.
	This results in the bold black matching in \Cref{fig:rho} as well as the set $F = \{\{a,c\},\{a,r\},\{r,s\},\{b,d\},\{d,t\}\}$.
	
	We want to compute $\matfn(m)$ for $m = \{a,b\}$, which occurs in the maximum weight matching $\{\{a,b\},\{r,s\},\{d,t\}\}$.
	In case that $w(a,c) = 2$, we are in the first case of the definition of $\matfn$, i.e., $a$ is matched with an edge above the matching threshold for matching $a$ with $b$.
	Hence, $\rho$ only considers the edges dominating $\{a,c\}$, the edge that is matching $a$ at the arrival of $b$. 
	The edges dominating $\{a,c\}$ are $\{a,r\}$ and $\{r,s\}$. Of these, only the former contains $a$, so we have  $\matfn(\{a,b\}) = \{\{a,r\}\}$.
	
	In case that $w(a,c) = 1$, we are in the first case of the definition of $\matfn$, i.e., $a$ is matched with an edge below the matching threshold for matching $a$ with $b$.
	In this case, $b$ is matched with $d$ (as in the first case), and the edges dominating $\{b,d\}$ and containing $b$ are also considered.
	Since $\{d,t\}$ is the only edge dominating  $\{b,d\}$ and does not contain $b$, we have that $\matfn(\{a,b\}) = \{\{a,r\},\{b,d\}\}$.
	
	Additionally, to illustrate the remaining proof, it is helpful to consider the example in \Cref{fig:ggdy} from the first part of the proof.
	There, $F = \{\{a_i,a_{i+1}\}\colon 0\le i \le k\}$.
	Moreover, the maximum weight matching is $\genmat^* = \{\{a_i,b_i\}\colon 0\le i\le k+1\}$.
	For $0\le i\le k$ we have $\matfn(\{a_i,b_i\}) = \{\{a_i,a_{i+1}\}\}$ according to the first case in the definition of $\matfn$.
	In addition, $\matfn(\{a_{k+1},b_{k+1}\}) = \{\{a_k,a_{k+1}\}\}$, also according to the first case in the definition of $\matfn$.
	
	\begin{claim}\label{cl:weightcorespondence}
		For all $m\in \genmat^*$, it holds that $w(m)\le \left(1 + \frac{\sqrt{2}}{2}\right)w(\matfn(m))$.
	\end{claim}
	\begin{claimproof}
		Let $m =\{a,b\}\in \genmat^*$ and assume that $b$ arrives after $a$. 
		Assume first that, at the arrival of $b$, $a$ is already matched with $c$ and $\left(1 + \frac{\sqrt{2}}{2}\right) w(a,c) > w(a,b)$.
		Then, for every edge $f\in F$ with $f\succsim_R \{a,c\}$ and $a\in f$, there exists a sequence of edges $e_0,\dots, e_j$ such that $e_0 = \{a, c\}$, $e_j = f$ and for $i\in [j]$, $e_i$ has replaced $e_{i-1}$.
		Hence, $w(\{a, c\}) = w(e_0) \le w(e_1)\le \dots \le w(e_j) = w(f)$.
		It follows that  $w(m) < \left(1 + \frac{\sqrt{2}}{2}\right) w(a,c) \le \left(1 + \frac{\sqrt{2}}{2}\right) w(\matfn(m))$.
		
		Assume next that, at the arrival of $b$, $a$ is already matched with $c$ and $\left(1 + \frac{\sqrt{2}}{2}\right) w(a,c) \le w(a,b)$.
		Then, $b$ will be matched with an agent $d$ such that $w(b,d) \ge w(a,b) - w(a,c)$.
		Note that this inequality also holds if $d$ is already matched to an agent $d'$ as then $w(b, d) \ge w(b,d) - w(d, d') \ge w(a, b) - w(a, c)$, where the latter inequality is due to the matching rule of $\ggdy$.
		As in the first case, we obtain $w(a,c)\le w(\max_{\succ_R}\{f\in F\colon f\succsim_R \{a,c\}, a\in f\})$ and $w(b,d)\le w(\max_{\succ_R}\{f\in F\colon f\succsim_R \{b,d\}, b\in f\})$.
		
		Now, if $d \in \{a,c\}$, then $\{a,b\}\precsim_R f$ for all $f\in \matfn(m)$ and $w(m) = w(a,b)\le w(\matfn(m))$.
		Otherwise, i.e., if $d\neq a$ and $d\neq c$, then $\{a,c\}\cap \{b,d\} = \emptyset$ and, therefore, also $\max_{\succ_R}\{f\in F\colon f\succsim_R \{a,c\}, a\in f\}\neq \max_{\succ_R}\{f\in F\colon f\succsim_R \{b,d\}, b\in f\}$.
		Hence, $w(m) \le w(b,d) + w(a,c) \le w(\matfn(m)) \le \left(1 + \frac{\sqrt{2}}{2}\right) w(\matfn(m))$.
		
		Finally, assume that $a$ is unmatched. 
		Then, $b$ will be matched with an agent $d$ such that $w(b,d) \ge w(a,b)$.
		As before, $w(b,d)\le w(\max_{\succ_R}\{f\in F\colon f\succsim_R \{b,d\}, b\in f\})$, and we conclude $w(m) \le w(b,d) \le w(\matfn(m)) \le \left(1 + \frac{\sqrt{2}}{2}\right) w(\matfn(m))$.
		This completes the proof of the claim.
	\end{claimproof}
	
	Next, we bound the number of edges that map to the same edge.
	Given an edge $e\in F$, define $\matfn^{-1}(e) := \{m\in \genmat^*\colon e \in \matfn(m)\}$.
	\begin{claim}\label{cl:fewmatches}
		If $e\in \genmat$, then $|\matfn^{-1}(e)|\le 2$.
		If $e\in F\setminus \genmat$, then $|\matfn^{-1}(e)|\le 1$.
	\end{claim}
	\begin{claimproof}
		By definition of $\matfn$, for every $m\in \genmat^*$ and $m'\in \matfn(m)$, $m'\cap m\neq \emptyset$.
		Let $e\in F$.
		Since $\genmat^*$ is a matching, there can be at most one edge in $\genmat^*$ with a non-empty intersection with each of the two endpoints of $e$.
		Hence, $|\matfn^{-1}(e)|\le 2$ and the first part of the claim holds.
		
		Now, let $e = \{a,b\}\in F\setminus \genmat$ and assume for contradiction that there exist edges $m,m'\in \genmat^*$ with $m\neq m'$, $e\cap m = \{a\}$, $e\cap m' = \{b\}$, and $e \in \matfn(m) \cap \matfn(m')$.
		Since $e\notin \genmat$, the edge $e$ was replaced by another edge $e'$ during the algorithm.
		By definition of $\ggdy$, it holds that $e\cap e'\neq \emptyset$, say $e\cap e' = \{a\}$.
		Moreover, $e'\succ_R e$, contradicting that $e\in \matfn(m)$.
		Hence, $|\matfn^{-1}(e)|\le 1$.
	\end{claimproof}
	
	It remains to bound the weight of the replaced edges. 
	For this, we introduce the following notation.
	For $e\in \genmat$, define $F_e = \{e'\in F\colon e\succ_R e'\}$.
	\begin{claim}\label{cl:weightchains}
		For all $e\in \genmat$, it holds that $w(F_e)\le \sqrt{2} w(e)$.
	\end{claim}
	\begin{claimproof}
		Let $e\in \genmat$.
		Since each edge can only replace a single other edge, $F_e$ is of the form $\{e_1,\dots, e_j\}$ for some $j\ge 0$ such that, for all $i\in [j-1]$, $e_i$ has replaced $e_{i+1}$, and $e$ has replaced $e_1$.
		Moreover, since a replacement only happens upon a sufficiently large weight improvement, we know that for all $i\in [j-1]$, $w(e_i)\ge \left(1 + \frac{\sqrt{2}}{2}\right)w(e_{i+1})$, and $w(e) \ge \left(1 + \frac{\sqrt{2}}{2}\right)w(e_1)$.
		Hence, for all $i\in [j]$, it holds that $w(e_i) \le \left(1+ \frac{\sqrt{2}}{2}\right)^{-i}w(e)$.
		Finally, it follows that 
		\begin{equation*}
			w(F_e) = \sum_{i=1}^jw(e_i) \le w(e) \sum_{i=1}^j \left(1 + \frac{\sqrt{2}}{2}\right)^{-i} \overset{\text{Eq.~(\ref{eq:partial_sum_bound}})}{=} w(e) \sqrt{2} \left(1 - \left(1 + \frac{\sqrt{2}}{2}\right)^{-j}\right) \le \sqrt{2} w(e)\text.
		\end{equation*}
		
	\end{claimproof}
	
	Combining all three claims, we compute
	\begin{align*}
		w(\genmat^*) & = \sum_{m\in \genmat^*}w(m) \overset{\textnormal{\Cref{cl:weightcorespondence}}}{\le} \sum_{m\in \genmat^*}\left(1 + \frac{\sqrt{2}}{2}\right)w(\matfn(m))
		 \overset{\textnormal{\Cref{cl:fewmatches}}}{\le} 2\left(1 + \frac{\sqrt{2}}{2}\right)\sum_{e\in \genmat}w(e) + \left(1 + \frac{\sqrt{2}}{2}\right)\sum_{e\in F\setminus \genmat}w(e)\\
		& \overset{\textnormal{\Cref{cl:weightchains}}}{\le} 2\left(1 + \frac{\sqrt{2}}{2}\right)\sum_{e\in \genmat}w(e) + \sqrt{2}\left(1 + \frac{\sqrt{2}}{2}\right)\sum_{e\in \genmat}w(e) \\
		 &= \left(2\left(1 + \frac{\sqrt{2}}{2}\right) + \sqrt{2}\left(1 + \frac{\sqrt{2}}{2}\right)\right)w(\genmat) 
		 = \left(3 + 2\sqrt{2}\right)w(\genmat)\text.
	\end{align*}
	This completes the proof.
\end{proof}

\section{Binomial coefficient identities and inequalities}\label{app:binom}

This section presents proofs for binomial identities and inequalities required for our main theorems.
Before we show the proofs, we briefly revisit classical identities for the binomial coefficient that we apply.
If the reader is unfamiliar with the topic, we recommend the book by \citet[][Chapter~5]{GKP94a}.
We also base our nomenclature of identities on this book.
Note that many of these identities also hold for negative and real values.
Still, for the ease of exposition, we restrict ourselves to nonnegative integers as this case is sufficient for our proofs.

Let $r, s, k, n \in \mathbb{N}_0$.
Then, the following binomial coefficient identities hold.
The \emph{factorial expansion} is $\binom{n}{k} = \frac{n!}{k!(n - k)!}$ if $n \ge k \ge 0$.
Next, by \emph{symmetry} $\binom{n}{k} = \binom{n}{n - k}$.
The \emph{absorption} or \emph{extraction} identity is $\binom{r}{k} = \frac{r}{k}\binom{r - 1}{k - 1}$ for $k \ne 0$.
We need the equivalent identity $\frac{1}{r + 1}\binom{r + 1}{k + 1} = \frac{1}{k + 1}\binom{r}{k}$ instead that holds for all $k \ge 0$.
The \emph{addition} or \emph{induction} identity is $\binom{r}{k} = \binom{r - 1}{k} + \binom{r - 1}{k - 1}$.
The \emph{Vandermonde convolution} is
\begin{equation}\label{eq:chu_vandermonde}
	\sum_{j = 0}^n \binom{r}{j}\binom{s}{n - j} = \binom{r + s}{n}\text.
\end{equation}

Furthermore, the \emph{parallel summation} is
\begin{equation}\label{eq:parallel_summation}
	\sum_{j = 0}^n \binom{r + j}{j} = \binom{r + n + 1}{n}\text.
\end{equation}

We now show proofs for the two equalities we used in \Cref{thm:wgdy}.

\begin{lemma}\label{lem:inner}
	For $k \ge i \ge 0$, it holds that
	\begin{equation*}
		\sum_{j = 0}^{i}  \frac{\binom{i}{j}}{\binom{k}{j}} = \frac{k + 1}{k + 1 - i}\text.
	\end{equation*}
\end{lemma}

\begin{proof}
	Let $k \ge i \ge 0$.
	We compute,
	\begin{equation*}
		\sum_{j = 0}^{i}  \frac{\binom{i}{j}}{\binom{k}{j}} = \frac{\binom{k}{i}}{\binom{k}{i}}\sum_{j = 0}^{i} \frac{i!}{j!(i-j)!}\frac{j!(k-j)!}{k!} =  \frac{1}{\binom{k}{i}}\sum_{j = 0}^{i}\frac{(k-j)!i!}{(i-j)!k!}\frac{k!}{i!(k - i)!} = \frac{1}{\binom{k}{i}}\sum_{j = 0}^{i} \frac{(k - j)!}{(i - j)!(k - i)!} = \frac{1}{\binom{k}{i}}\sum_{j = 0}^{i} \binom{k - j}{i - j}\text.
	\end{equation*}

	Next, we change the summation twice to bring it into a form where we can use \Cref{eq:parallel_summation}.
	Then,
	\begin{equation*}
		\frac{1}{\binom{k}{i}}\sum_{j = 0}^{i} \binom{k - j}{i - j} = \frac{1}{\binom{k}{i}}\sum_{j = i}^{0} \binom{k - (i - j)}{i - (i - j)} = \frac{1}{\binom{k}{i}}\sum_{j = i}^{0} \binom{k - i + j}{j}= \frac{1}{\binom{k}{i}}\sum_{j = 0}^{i} \binom{k - i + j}{j}\text. 
	\end{equation*}
	
	We now apply \Cref{eq:parallel_summation} for the case that $r = k - i$ and $n = i$, i.e., $\sum_{j = 0}^{i}\binom{k - i + j}{j} = \binom{k + 1}{i}$.
	This yields
	\begin{equation*}
		\frac{1}{\binom{k}{i}}\sum_{j = 0}^{i} \binom{k - i + j}{j}  = \frac{\binom{k + 1}{i}}{\binom{k}{i}}\text.
	\end{equation*}
	
	Simplifying this expression proves the lemma
	\begin{equation*}
		\frac{\binom{k + 1}{i}}{\binom{k}{i}} = \frac{(k + 1)!}{i!(k + 1 - i)!}\frac{i!(k - i)!}{k!} = \frac{k + 1}{k + 1 - i}\text.\qedhere
	\end{equation*} 
\end{proof}

Next, we show the proof for the second equality used in the proof of \Cref{thm:wgdy}.

\begin{lemma}\label{lem:outer}
	For $k \in \mathbb{N}$, it holds that
	\begin{equation*}
		 \sum_{i = 0}^{k} \frac{\binom{k}{i}\binom{k}{k - i}}{\binom{2k}{k}} \frac{1}{k + 1 - i} 
		 = \frac{2k + 1}{(k + 1)^2}\text.
	\end{equation*}
\end{lemma}
\begin{proof}
	Let $k \in \mathbb{N}$.
	We apply symmetry and change the summation by replacing $k - i$ with $i$ to simplify the equation, i.e.,
	\begin{equation*}
		\sum_{i = 0}^{k} \frac{\binom{k}{i}\binom{k}{k - i}}{\binom{2k}{k}} \frac{1}{k + 1 - i} =  \sum_{i = 0}^{k} \frac{\binom{k}{k - i}^2}{\binom{2k}{k}} \frac{1}{k + 1 - i} = \sum_{i = 0}^{k} \frac{\binom{k}{i}^2}{\binom{2k}{k}} \frac{1}{i + 1}\text.
	\end{equation*}
	
	Next, we apply the extraction identity to make the fraction independent of the sum.
	Moreover, we apply the induction identity to receive two separate sums.
	We compute
	\begin{align*}
		\sum_{i = 0}^{k} \frac{\binom{k}{i}^2}{\binom{2k}{k}} \frac{1}{i + 1} &= \sum_{i = 0}^{k} \frac{\binom{k}{i}\binom{k+1}{i+1}}{\binom{2k}{k}}\frac{1}{k + 1} = \frac{1}{(k + 1)\binom{2k}{k}}  \sum_{i = 0}^{k} \binom{k}{i}\binom{k+1}{i+1}\\
		& = \frac{1}{(k + 1)\binom{2k}{k}}  \left(  \sum_{i = 0}^{k} \binom{k}{i}^2 + \sum_{i = 0}^{k} \binom{k}{i}\binom{k}{i+1} \right)\text.
	\end{align*}

	Now, we use the Vandermonde convolution twice.
	We use the case where $r = s = k$ in \Cref{eq:chu_vandermonde}.
	For the first application, we additionally have the case $n = k$, i.e., $\sum_{j = 0}^k \binom{k}{j}^2 = \binom{2k}{k}$, and for the second application we have the case $n = k - 1$, i.e., $\sum_{j = 0}^{k - 1} \binom{k}{j}\binom{k}{k - j - 1} = \binom{2k}{k - 1}$.
	The first identity can be applied directly.
	For the second, we first transform the sum, using symmetry and that its last term is $0$, i.e.,
	\begin{equation*}
		 \sum_{i = 0}^{k} \binom{k}{i}\binom{k}{i+1} =  \sum_{i = 0}^{k - 1} \binom{k}{i}\binom{k}{i+1} =  \sum_{i = 0}^{k - 1} \binom{k}{i}\binom{k}{k - i -1} = \binom{2k}{k - 1}\text.
	\end{equation*}

	We insert these simplifications back into our equation and once again apply the induction identity, i.e.,
	\begin{equation*}
		\frac{1}{(k + 1)\binom{2k}{k}}  \left(  \sum_{i = 0}^{k} \binom{k}{i}^2 + \sum_{i = 0}^{k} \binom{k}{i}\binom{k}{i+1} \right) = \frac{1}{(k + 1)\binom{2k}{k}}  \left( \binom{2k}{k} + \binom{2k}{k - 1} \right) = \frac{\binom{2k + 1}{k}}{(k + 1)\binom{2k}{k}}\text.
	\end{equation*}

	Finally, we simplify the equation to show the lemma
	\begin{equation*}
		\frac{\binom{2k + 1}{k}}{(k + 1)\binom{2k}{k}} = \frac{1}{k+1}\frac{(2k + 1)!}{k!(2k+1-k)!}\frac{k!(2k-k)!}{2k!} = \frac{2k + 1}{(k+1)^2}\text.\qedhere
	\end{equation*}
\end{proof}

We conclude with the proof of the inequality used in \Cref{thm:IEW}.

\begin{lemma}\label{lem:crazy}
	For $i, k \in \mathbb{N}_0$ such that $k + 2 > 2^i$, it holds that
	\begin{equation*}
	\sum_{j = 0}^{2^i - 1} \frac{\binom{k}{j}\binom{k}{2^i - 1 - j}}{\binom{2k}{2^i-1}} \frac{1}{2^i - j} \le \frac{2}{2^i}\text.
	\end{equation*}
\end{lemma}
\begin{proof}
	Let $i, k \in \mathbb{N}_0$ such that $k + 2 > 2^i$.
	Let $r := 2^i - 1$. This simplifies the equation we need to show to
	\begin{equation*}
	\sum_{j = 0}^{r} \frac{\binom{k}{j}\binom{k}{r - j}}{\binom{2k}{r}} \frac{1}{r + 1 - j} \le \frac{2}{r + 1}\text.
	\end{equation*}

	We first apply the extraction identity and pull the fraction out of the sum.
	Then, we extend the sum with an additional term in preparation for the following step.
	We have
	\begin{equation*}
		\sum_{j = 0}^{r} \frac{\binom{k}{j}\binom{k}{r - j}}{\binom{2k}{r}} \frac{1}{r + 1 - j} = \frac{1}{\binom{2k}{r}}\sum_{j = 0}^r \binom{k}{j}\binom{k + 1}{r - j + 1} \frac{1}{k + 1} \le \frac{1}{(k + 1)\binom{2k}{r}}\sum_{j = 0}^{r + 1} \binom{k}{j}\binom{k + 1}{r - j + 1}\text.
	\end{equation*}
	
	Now, we use the Vandermonde convolution to get rid of the sum.
	We use \Cref{eq:chu_vandermonde} for the case where  we identify $r$ with $k$, $s$ with $k + 1$, and $n$ with $r + 1$, i.e., $\sum_{j = 0}^{r+1} \binom{k}{j}\binom{k + 1}{r - j + 1} = \binom{2k + 1}{r+1}$. 
	\begin{equation*}
		\frac{1}{(k + 1)\binom{2k}{r}}\sum_{j = 0}^{r + 1} \binom{k}{j}\binom{k + 1}{r - j + 1} = \frac{1}{(k + 1)\binom{2k}{r}}\binom{2k + 1}{r + 1}
	\end{equation*}
	
	Finally, the lemma follows from simplifying the term, i.e.,
	\begin{equation*}
		\frac{\binom{2k + 1}{r + 1}}{(k + 1)\binom{2k}{r}} = \frac{1}{k + 1}\frac{(2k + 1)!}{(r + 1)!(2k - r)!}\frac{r!(2k - r)!}{2k!} = \frac{2k + 1}{(k + 1)(r + 1)} \le \frac{2k + 2}{(k + 1)(r + 1)} = \frac{2}{r + 1} = \frac{2}{2^i}\text.\qedhere
	\end{equation*}
	
\end{proof}

\end{document}